\documentclass{eptcs}
\usepackage[utf8]{inputenc}
 
\usepackage{xargs}
\usepackage[capitalise]{cleveref}

\usepackage[nomargin,multiuser,inline]{fixme} 
\fxusetheme{color}
\FXRegisterAuthor{rg}{arg}{\color{blue}   {\underline{rG}}}
\FXRegisterAuthor{eM}{aeM}{\color{orange} {\underline{eM}}}

\usepackage{mathtools}
\usepackage{amsmath}
\usepackage{amsthm}
\newtheorem*{remark}{Remark}

\usepackage{amssymb}
\usepackage{caption}                % Needed by subcaption
\usepackage{cmll} % \bigwith symbol
\usepackage{colonequals}            % More symbols (colon equalities etc.)
\usepackage{environ}                % Provides NewEnviron w/ eager expansion
\usepackage{etex,etoolbox}          % Needed by theorem-appendix.tex
\usepackage{graphicx}               % For \scalebox etc.

\usepackage[protrusion=true,expansion=true]{microtype} % Better typography
\usepackage{multirow}               % Table multirow/multicolumn support
\usepackage{nicefrac}
\usepackage{proof}                  % Logic proofs
\usepackage{bussproofs}
\EnableBpAbbreviations

\usepackage{rotating}               % For \sidewaysfigure etc.
\usepackage{subcaption}             % Support for subfigures
\usepackage{thmtools, thm-restate, thm-autoref}  % Advanced theorem handling
\usepackage{xifthen}                % for conditional commands
\usepackage{appendix}               % for correct numeration of statements

\usepackage{amsbsy} % for bold math symbols
\usepackage{stmaryrd}
 \usepackage{pdfsync}

\usepackage[usenames,dvipsnames]{xcolor} % must be loaded first
\usepackage{tikz}                   % Graphs
\usetikzlibrary{arrows,shapes,automata,positioning,decorations.markings}

\usepackage{fixltx2e} % Text mode subscript and superscript
\usepackage{xspace} % Smart spaces after commands

\usepackage[nomargin,inline]{fixme} % Simplified management of FIXME's
\fxusetheme{color}

\usepackage{mathtools} % \xRightArrow and stuff - load late to avoid conflicts

\hypersetup{final} % Needed to generate hyperlinks even in draft mode

\def\colorPtp{\color{black}}
\def\colorSbj{\color{black}}
\def\colorAct{\color{black}}
\def\colorNode{\color{black}}
\def\colorFun{\color{black}}
\def\colorR{\color{black}}
\def\colorE{\color{black}}

\newcommand{\quo}[1]{\lq\lq {#1}\rq\rq}

\newcommand{\msg}[1][m]{\mathsf{#1}}
\newcommand{\msgset}{\mathsf{M}}
\newcommand{\chset}{\mathsf{C}}
\newcommand{\eset}{\mathsf{E}}
\newcommand{\gset}{\mathcal{G}}
\newcommand{\ptp}[1][A]{\ensuremath{\mathsf{#1}}}

\newcommand{\subject}[1][\ae]{\mathsf{\colorSbj{sbj(#1)}}}

\newcommandx{\common}[3][1=\ptp,2={\aR},3={\aR'},usedefault=@]{f_{#1}}
\newcommandx{\opair}[2][1={\ae},2={\ae'},usedefault=@]{\langle {#1},{#2} \rangle}
\newcommandx{\hopair}[2][1={\aE},2={\aE'},usedefault=@]{\llparenthesis\, {#1},{#2}\, \rrparenthesis}
\newcommandx{\nmerge}[2][1={i},2={},usedefault=@]{\mu_{#2}\ifempty{#1}{}{(\gname[{#1}])}}
\newcommandx{\wb}[2][1={\aG},2={\aG'},usedefault=@]{wb({#1}, {#2})}
\newcommandx{\widx}[2][1={\aW},2={i},usedefault=@]{{#1}[{#2}]}
\newcommandx{\outop}[2][1=\gname,2={}]{!^{{#1}{#2}}}
\newcommandx{\inop}[2][1=\gname,2={}]{?^{{#1}{#2}}}
\newcommandx{\aout}[5][1=\p,2={\q},3=\gname,4=\msg,5={},usedefault=@]{
  {#1}{#5}{#2}{#5}\outop[{{#3}{#5}}]{#4}{#5}
}
\newcommandx{\ain}[5][1=\p,2={\q},3=\gname,4=\msg,5={},usedefault=@]{
  {#1}{#5} {#2}{#5} \inop[{{#3}{#5}}] {#4}{#5}
}
\newcommandx{\adep}[1][1={}]{
  \langle \aout[@][@][@][@][{#1}], \ain[@][@][@][@][{#1}] \rangle
}

\newcommandx{\hproj}[2][1=\aR, 2=\ptp, usedefault=@]{
  \ifempty{#1}{}{{#1}}\ifempty{#2}{}{{^{\tiny @{#2}}}}
}
\newcommandx{\eproj}[2][1=\aE,2=\ptp, usedefault=@]{
  {{#1}}\ifempty{#2}{}{{^{\tiny @{#2}}}}
}

\newcommand{\aM}{M}
\newcommandx{\cm}[2][1=\ptp, 2=\aM]{{#2}_{#1}}
\newcommand{\p}{\ptp}
\newcommand{\q}{{\ptp[B]}}
\newcommandx{\achan}[2][1=A,2=B,usedefault=@]{\ptp[#1]\ptp[#2]}

\newcommand{\ptpset}{\mathsf{\colorPtp{P}}}

\newcommand{\pset}{\ptpset}

\newcommand{\tset}{\to}

\newcommand{\csconf}[2]{\conf{\vec{#1} \ ; \ \vec{#2}}}
\newcommand{\conf}[1]{\langle #1 \rangle}
\newcommand{\chanset}{\chset}
\newcommand{\TRANSS}[1]{{\xRightarrow{\raisebox{-.3ex}[0pt][0pt]{\scriptsize $#1$} }}}
\newcommand{\RS}{\mathsf{RS}}
\newcommand{\NUL}{\varepsilon}

\newcommand{\trans}[2][{}]{\,\xrightarrow{#2}_{#1}\,}

\newcommand{\aG}{\mathsf{G}}
\newcommand{\gsink}{\circledcirc}
\newcommand{\gsource}{\circ}
\newcommand{\gvertex}{\bullet}
\newcommand{\gseqop}{;}
\newcommand{\gparop}{|}
\newcommand{\gchoop}{+}
\newcommand{\gname}[1][i]{{\colorNode{\scriptsize\textsf{#1}}}}

\newcommandx{\gnode}[2][1=\gname,2=\gint,usedefault=@]{
  {\ifempty{#1}{}{\colorNode{#1\colon}}} {#2}
}
\newcommand{\nodeset}{\mathsf{\colorNode{K}}}
\newcommand{\gempty}{\mathbf{0}}
\newcommandx{\gint}[4][1=\gname,2=\ptp,3=\msg,4={\ptp[B]},usedefault=@]{
  \gnode[{#1}][{#2} \xrightarrow{\msg[{#3}]} {#4}]
}
\newcommandx{\gseq}[3][1=\gname,2={\aG},3={\aG'},usedefault=@]{
  \gnode[#1][{#2} \gseqop {#3}]
}
\newcommandx{\gpar}[3][1=\gname,2={\aG},3={\aG'},usedefault=@]{
  \gnode[#1][\ifempty{#1}{{#2} \gparop {#3}}{\!\!({#2} \gparop {#3})}]}
\newcommandx{\gcho}[3][1=\gname,2={\aG},3={\aG'},usedefault=@]{
  \gnode[#1][\ifempty{#1}{{#2} \gchoop {#3}}{\!\!({#2} \gchoop {#3})}]
}

\newcommandx{\gsem}[2][1={\aG},2={},usedefault=@]{[\![ {#1} ]\!]_{#2}}
\newcommandx{\rbot}{\perp}
\newcommandx{\rtrs}[1][1={\aR},usedefault=@]{{#1}^{\star}}
\newcommandx{\gord}[1][1={\aG},usedefault=@]{<_{#1}}
\newcommandx{\gordeq}[1][1={\aG},usedefault=@]{\leq_{#1}}

\newcommandx{\rlang}{\mathcal{L}}
\newcommandx{\efst}[2]{\textsf{cs}_{#1}\ifempty{#2}{}{({#2})}}
\newcommand{\enode}[1]{\colorAct{\textsf{cp}}\ifempty{#1}{}{({#1})}}
\newcommand{\eact}[1]{{\colorAct{\textsf{act}}}\ifempty{#1}{}{({#1})}}
\newcommandx{\aW}{w}

\newcommand{\gfun}[1]{\ensuremath{\mathsf{\colorFun #1}}}
\newcommand{\rmax}[1][\aR]{\gfun{max}\,{#1}}
\newcommand{\rmin}[1][\aR]{\gfun{min}\,{#1}}
\newcommand{\rMAX}[1][\aR]{\gfun{lst}\,{#1}}
\newcommand{\rMIN}[1][\aR]{\gfun{fst}\,{#1}}

\newcommandx{\rseq}[2][1=\aG,2={\aG'},usedefault=@]{\gfun{seq}({#1},{#2})}

\newcommand{\aR}[1][R]{{\colorR{#1}}}
\renewcommand{\ae}[1][e]{{\colorE{#1}}}
\newcommand{\qst}{\, : \, }

\newcommandx{\gproj}[2][1=\aG,2=\ptp]{{#1}\downarrow_{#2}}
\newcommandx{\cinit}[1][1={q_0},usedefault=@]{{#1}}
\newcommandx{\cfinal}[1][1={q_e},usedefault=@]{{#1}}
\newcommand{\ctr}[3]{
  \tikz{
    \node at (0,0) (i) {};
    \node[shape=circle,draw,inner sep=1pt] at (0.7,0) (s) {$#1$};
    \node[shape=circle,draw,inner sep=1pt, right=of s] at (1.5,0)  (t) {$#2$};
    \node at (3.1,0) (f) {};
    \draw[->] (i) -- (s);
    \draw[->] (s) -- (t) node [midway,above]{$#3$};
    \draw[->] (t) -- (f);
  }
}
\newcommand{\ccup}{\ \cup\ }
\newcommand{\ccap}{\ \cap\ }
\newcommand{\cprod}{\ \times\ }
\newcommandx{\geproj}[4][1=\aG,2=\ptp,3=\cinit,4=\cfinal,usedefault=@]{
  {#1}\downarrow_{#2}^{{#3},{#4}}
}

\newcommand*{\StrikeThruDistance}{0.15cm}%
\tikzset{strike thru arrow/.style={
    decoration={markings, mark=at position 0.5 with {
        \draw [blue, thick,-] 
            ++ (-\StrikeThruDistance,-\StrikeThruDistance) 
            -- ( \StrikeThruDistance, \StrikeThruDistance);}
    },
    postaction={decorate},
}}

\newcommandx{\ich}[1][1={\aG},usedefault=@]{{#1}^{\oplus}}
\newcommandx{\ichedges}[2][1={\aG},2={\gname},usedefault=@]{{#1}^{\oplus}({#2})}
\newcommandx{\parts}[1]{2^{#1}}
\newcommandx{\actch}{c}
\newcommandx{\soundactch}[2][1={\aG},2={\actch},usedefault=@]{{#1} \,\circledR\, {#2}}
\newcommandx{\rOnActch}[2][1={\aG},2={\actch},usedefault=@]{{#1} \setminus {#2}}
\newcommandx{\rOnActchClean}[2][1={\aG},2={\actch},usedefault=@]{{#1} \circledR {#2}}
\newcommandx{\rAllEvents}[1][1={\aG},usedefault=@]{\mathit{dom}(#1)}
\newcommand{\aE}{{\tilde \ae}}
\newcommandx{\hyedge}[1]{\{#1\}}

\newcommandx{\rdiv}[2][1=\gcho,2=\ptp,usedefault=@]{
  \gfun{div}_{#2}(#1)
}

\newcommandx{\rrdiv}[5][1={\aG},2={\aG'},3={\aE},4={\aE'},5=\ptp,usedefault=@]{
  \gfun{div}^{#3,#4}_{#5}(#1,#2)
}

\newcommand{\aCM}{M}
\newcommand{\aQ}{Q}
\newcommandx{\aQzero}[1][1=,usedefault=@]{
  {\ifempty{#1}{q_0}{q_{q#1}}}
}
\newcommand{\aTrs}{\tset}
\newcommand{\aCS}{S}
\newcommand{\aConf}{s}

\newcommand{\gatedistancein}{3pt}
\newcommand{\gatedistanceinand}{2pt}

\tikzset{
  src/.style={draw,circle,fill=white,
    minimum size=2mm,
    inner sep=0pt},
  sink/.style={draw,circle,double,fill=white,
    minimum size=1.5mm,
    inner sep=0pt},
  node/.style={draw,circle,fill=black,
    minimum size=2mm,
    inner sep=0pt},
  block/.style = {rectangle, draw=gray, align=center, fill=orange!25, rounded corners=0.1cm,
    minimum size=5mm, inner sep=2pt},
  prenode/.style = {minimum size=9pt,inner sep=2pt, font=\Large},
  bblock/.style = {rectangle, draw=blue!50, opacity=.5, line width=1pt, align=center, fill=white, rounded corners=0.1cm,
    minimum size=7mm, inner sep=2pt},
  prenode/.style = {minimum size=9pt,inner sep=2pt, font=\Large},
  agate/.style={draw, rectangle,
    minimum size=3mm,
    inner sep=0pt,
    fill=orange!25,
    postaction={path picture={% 
        \draw[red]
        ([yshift=\gatedistanceinand]path picture bounding box.south) --
        ([yshift=-\gatedistanceinand]path picture bounding box.north) ;}}},
  ogate/.style = {
    diamond, draw, fill=orange!25,
    minimum size=4mm,
    inner sep=0pt,
    postaction={path picture={% 
        \draw[red]
        ([yshift=\gatedistancein]path picture bounding box.south) -- ([yshift=-\gatedistancein]path picture bounding box.north)
        ([xshift=-\gatedistancein]path picture bounding box.east) -- ([xshift=\gatedistancein]path picture bounding box.west)
        ;}}},
  anygate/.style = {circle, draw, fill=white,
    minimum size=4mm,
    inner sep=0pt,
    postaction={path picture={% 
        \draw[black]
        ([xshift=-\gatedistancein,yshift=\gatedistancein]path picture bounding box.south east) --
        ([xshift=\gatedistancein,yshift=-\gatedistancein]path picture bounding box.north west)
        ([xshift=-\gatedistancein,yshift=-\gatedistancein]path picture bounding box.north east) --
        ([xshift=\gatedistancein,yshift=\gatedistancein]path picture bounding box.south west)
        ;}}},
  elli/.style = {draw,densely dotted,-},
  line/.style = {draw,->, rounded corners=0.07cm,>=latex},
  nline/.style = {draw,semithick, ->},
  pline/.style = {draw,->,>=latex},
  node distance=1cm and 0.7cm,
  baseline=(current  bounding  box.center),
}

\newcommand{\ifempty}[3]{%
  \ifthenelse{\isempty{#1}}{#2}{#3}%
}

\newcommand{\bnfdef}{\ ::=\ }
\newcommand{\bnfmid}{\ \big|\ }

\newcommand{\hidden}[1]{}

\newcommand{\st}{\ \big|\ }

\def\vec{\mathaccent"017E }
\renewcommand{\vec}[1]{\ensuremath{\textit{\textbf{#1}}}}

\let\greekgamma\gamma
\def\contrColor{\color{Plum}}
\newcommand{\contrFmt}[1]{{\contrColor{#1}}}

\newcommand{\codefont}{\fontsize{10}{10}\selectfont}
\newcommand{\code}[1]{{\tt\codefont {#1}}}

\renewcommand{\gamma}[1][]{\mathord{\contrFmt{\greekgamma}_{\contrFmt{#1}}}}

\def\cocoColor{\color{MidnightBlue}}
\newcommand{\cocoFmt}[1]{{\cocoColor{\code{#1}}}}

\let\greektau\tau
\renewcommand{\tau}{\cocoFmt{\greektau}}

\NewEnviron{restheorem}[2][]
  {\begin{restatable}[#1]{thm}{#2}%
    \label{#2}
    \BODY
   \end{restatable}%
   \begin{proof}%
     See~\vref{#2-proof}.%
   \end{proof}}
\NewEnviron{reslemma}[2][]
  {\begin{restatable}[#1]{lem}{#2}%
    \label{#2}
    \BODY
   \end{restatable}%
   \begin{proof}%
     See~\vref{#2-proof}.%
   \end{proof}}

  {}

\newtheorem{theorem}{Theorem}

\title{An Abstract Semantics of the Global View of Choreographies\thanks{The authors are grateful to the reviewers of ICE for the helpful comments and
discussions on the forum. This work has been partially supported by
COST Action IC1201 (Behavioural Types for Reliable Large-Scale Software Systems, BETTY).}}
\def\titlerunning{An Abstract Semantics of the Global View of Choreographies}

\author{
  Roberto Guanciale
  \institute{KTH, Sweden}\\
  \email{robertog@kth.se}
  \and
  Emilio Tuosto
  \institute{University of Leicester, UK}\\
  \email{emilio@le.ac.uk}
}
\def\authorrunning{R. Guanciale \& E. Tuosto}

\begin{document}

\maketitle              % typeset the title of the contribution

\begin{abstract}
  We introduce an abstract semantics of the global view of
  choreographies. Our semantics is given in terms of pre-orders and
  can accommodate different lower level semantics. We discuss the
  adequacy of our model by considering its relation with communicating
  machines, that we use to formalise the local view. Interestingly,
  our framework seems to be more expressive than others where
  semantics of global views have been considered. This will be
  illustrated by discussing some interesting examples.
\end{abstract}

%% Documents sections follow
\section{Introduction}\label{intro:sec}
\paragraph{The problem}
Choreographies have been advocated as a suitable methodology for the
design and analysis of distributed applications.
Roughly, a choreography describes how two of more distributed
components coordinate with each other.
Of course, in a distributed setting this coordination has to happen
through exchange of messages.
Among the possible interpretations of what choreographies are
(see~\cite{bdft16} for a discussion and references), we embrace the
one suggested by W3C's~\cite{w3c:cho}:
\begin{quote}
  Using the Web Services Choreography specification, a contract
  containing a global definition of the common ordering conditions and
  constraints under which messages are exchanged, is produced that
  describes, from a \textbf{global viewpoint} [...] observable
  behaviour [...]. Each party can then use the \textbf{global
    definition} to build and test solutions that conform to it. The
  global specification is in turn realised by combination of the
  resulting \textbf{local systems} [...]
\end{quote}
This description conceptualises two views, a \textbf{global} and a
\textbf{local} one, which enable the relations represented by the
following diagram:
\begin{equation}
  \label{eq:chor:dia}
  \begin{tikzpicture}[node distance=1cm and 1cm, every text node part/.style={align=center}]
    \node at (0,0) (g) [text width = 1cm] {Global view};
    \node at (5,0) (l) {Local \\view};
    \node at (10,0) (c) {Local \\systems};
    \draw[->] (g) -- (l) node [midway,above]{projection};
    \draw[->] (c) -- (l) node [midway,above]{comply};
  \end{tikzpicture}
\end{equation}
where \lq projection\rq\ is an operation producing the local view from the
global one and \lq comply\rq\ verifies that the behaviour of each
components adhere with the one of the corresponding local
view.
(The \lq projection\rq\ arrow in \eqref{eq:chor:dia} may have an
\quo{inverse} one (cf.~\cite{lty15}), but this is immaterial here.)
For diagram~\eqref{eq:chor:dia} to make sense, precise semantics should
be fixed for the global and the local views.
The semantics of the latter is well understood: it directly emanates
from the adopted communication model.
In fact, the local view details how communications take place.
For instance, in a channel-based communication model, the local view
may specify what is the behaviour of each component in terms of its
send/receive actions.

What is instead \quo{the semantics of the global view}?
We investigate such question here. And, after making it more precise,
we propose a new semantic framework for global views and discuss its
advantages on existing frameworks.

\paragraph{A view of global views}
Although intriguing, the W3C description above, is not very
enlightening to understand what a global view is; basically it says
that a global view has to describe the observable behaviour from a
global viewpoint...a bit too much circularity for a definition!

We will consider global views as high level descriptions of systems
abstracting away some aspects in order to offer a \emph{holistic}
understanding of the communication behaviour of distributed
systems.
(We beg for the reader's patience: this is still vague, but will
become precise in the forthcoming sections.)
In a global view, components are not taken anymore in \emph{isolation}.
Rather they are \emph{specified together}, while \emph{forgetting
  some details}.
For us, this will mean to describe the protocol of interaction of a
systems in a way that is oblivious of how messages are actually
exchanged in the communication.
For instance, in our example based on channels, the global view may
abstract away from send/receive actions and use \emph{interactions}
as the unit  of coordination~\cite{chy07}.

The idea depicted in diagram~\eqref{eq:chor:dia} is beautiful.
To our best knowledge, it has been firstly formally pursued
in~\cite{honda16jacm} and later followed by others.
The main reason that makes attractive diagram~\eqref{eq:chor:dia} is
the interplay between global and local artefacts%
\footnote{%\begin{remark}
  % Allow us a terminological digression.
  %
  We will use the term \lq artefact\rq\ when referring to actual
  specifications embodying the global/local views.
  Such embodiments may assume various forms: types~\cite{honda16jacm},
  programs~\cite{ivan}, graphs and automata~\cite{lty15,dy12}, executable
  models~\cite{w3c:cho,bpel}, etc.
  Typically, the literature uses the (overloaded) word \lq model\rq\
  to refer to this flora of embodiments.
  We prefer the word \lq artefact\rq\ because it allows us to refer to
  different contexts and different abstraction levels without
  attaching yet another meaning to \lq model \rq.
}
as it fosters some of the best principles of computer science:
\begin{description}
\item[Separation of concerns] The \emph{intrinsic logic} of the
  distributed coordination is expressed in and analysed on global
  artefacts, while the local artefacts refine such logic at lower
  levels of abstraction.
\item[Modular software development life-cycle] The W3C description
  above yields a distinctive element of choreographies which makes
  them appealing (also to practitioners).
  Choreographies allow independent development: components can
  harmoniously interact if they are proven to comply with the local
  view.
  Global and local views yield the \quo{blueprints} of systems as a
  whole and of each component, respectively.
\item[Principled design] A choreographic framework orbits around the
  following implication:
\begin{center}
  \textbf{if} \emph{cond}(global artefact) \textbf{then}
  \emph{behave}(\emph{projection}(global artefact))
\end{center}
that is, proving that a correctness condition \emph{cond} holds on an
abstraction (the global artefacts) guarantees that the system is well
behaved, provided that the local artefacts are \quo{compiled} from the
global ones via a \emph{projection} operation that preserves
behaviour.
\end{description}

Therefore, providing good semantics for global artefacts is
worthwhile: it gives precise algorithms and establishes precise
relations between specifications of distributed systems (the global
artefacts) and their refinements (the local artefacts).
\\\noindent
\textbf{Outline \& Contributions}
We explain the advantages of defining an abstract semantics of global
views in \cref{mot:sec} and we give the syntax of our language of
global artefacts in \cref{cho:sec}.
\cref{hyper:sec} is a technical prelude; it introduces the  notion
of \emph{reflection}, which is crucial for our generalisation.
\cref{sem:sec} yields another contribution: our abstract semantics of
global artefacts. A first technical advantage of our semantics is
provided by the definition of \emph{well-branched} choices, explained
through some the illustrative examples of \cref{sem:sec}.
Our semantics is used in \cref{lang:sec} to identify all licit traces
of a choreography, thus making it possible to precisely characterise
the behaviour expected by the specification.
\cref{proj:sec} first recalls the communicating finite state machines
(that are used to formalise the local behaviours) and then defines the
projection of global artefacts on communicating machines.
The main technical results establish that well-branched choreographies
are deadlock free (\cref{thm:deadlock}) and that the executions specified
by the global view contain those of its projections (\cref{thm:lang})
operation and shows that the local behaviours comply with the ones of
the global specification.  Concluding remarks are in \cref{conc:sec}.

% Besides solving (or at least mitigating) the drawbacks of existing
% approaches considered above, our model also caters for different
% semantics attainable by simply stating different dependencies among
% communication actions.

%%% Local Variables:
%%% mode: latex
%%% TeX-master: "main"
%%% End:

%  LocalWords:  Choreographies embodiments automata choreographies

\section{Why going abstract?}\label{mot:sec}
As said, many authors have adopted the idea in diagram
\eqref{eq:chor:dia} and several semantics of (models of) global views
have been introduced.
We distinguish two broad classes.
\begin{remark}
  We mention a tiny portion of the literature in way of example; no
  claim of exhaustiveness.
\end{remark}

The largest class is possibly the one that includes the seminal work
on global types~\cite{honda16jacm}.
The idea is that the semantics of global artefacts (embodied by
global types in~\cite{honda16jacm}) is given in terms of the semantics
of their local artefacts via a suitable projection operation.
In the case of global types, the projection yields local types,
that are process algebras equipped with an operational semantics.
This approach is ubiquitous in the literature based on behavioural
types and it has also been adopted in~\cite{lty15} where global
artefacts are global graphs~\cite{dy12} and local artefacts are
communicating machines~\cite{bz83}.

In the other class, the semantics of global views is defined
explicitly.
For instance, in~\cite{DBLP:journals/corr/abs-1203-0780} an
operational semantics is defined while in~\cite{bmt14} a trace-based
semantics is given.
In both cases, the idea is to \quo{split} the interactions in the
global view into its constituent send/receive actions.
In this category we also put approaches like~\cite{ivan} where
global artefacts become \emph{global programs} with an operational
semantics.

The classes above contain perfectly reasonable approaches,
from a theoretical perspective.
After all, we just need a semantics for the global view; whatever
\quo{fits} with the semantics of the local view would do.
We argue however that making the semantics of the global view
a \emph{dependent variable} of the semantics of the local one
brings in some issues that we now briefly discuss.

Firstly, several (syntactic) restrictions are usually necessary
in order
to rule out choreographies that \quo{do not make sense}.
Such restrictions may be innocuous (as for instance the requirement
that the components involved in two sequentially consecutive
interactions cannot be disjoint), but they could also limit the
expressiveness of the language at hand (for instance, languages
featuring the parallel composition of global artefacts do not allow
components involved in more than one parallel thread).

Secondly, and more crucially, the semantics of global views proposed
so far appear to be \quo{too concrete}.
As a matter of fact, this spoils the beauty of the interplay between
global and local views.
All the semantics of the global view that we are aware of basically
mirror quite closely the one of the local view.
This means that to understand a global artefact one has to look
at (or think in terms of) the corresponding local artefacts.
This is not only difficult to do, but also undesirable.
For instance, designers have to know/fix low level details at early
stages of the development and cannot really compare different global
artefacts with each other without considering the local artefacts;
this makes it hard to e.g., take design decisions at the abstract
level.

So, what about giving a semantics of the global view
\emph{independently} of the one of the local view?
This is what we do here.
We define a new semantics of global views that makes very few
assumptions on how messages are exchanged at lower levels.
Conceptually this is easy to achieve.
We fix a specification language of global artefacts and we interpret a
specification as a set of \quo{minimal and natural} causal
dependencies among the messages.
We then define when a global artefact is sound, namely when its causal
dependencies are consistent so that they are amenable to be executed
distributively by some local artefacts, regardless of the underlying
message passing semantics.

We illustrate the advantages of our approach by adopting a rather
liberal language of global artefacts inspired by global
graphs~\cite{dy12}.
We then show the relation of such language on a local view featuring
local artefacts as communicating machines~\cite{bz83}.

%%% Local Variables: 
%%% mode: latex
%%% TeX-master: "main"
%%% End: 

%  LocalWords:  choreographies

\section{Global views as Graphs }\label{cho:sec}
Let $\ptpset$ be a set of \emph{participants} (ranged over by $\ptp$,
$\ptp[B]$, etc.), $\msgset$ a set of \emph{messages} (ranged over by
$\msg$, $\msg[x]$, etc.), and $\nodeset$ a set of \emph{control
  points} (ranged over by $\gname$, $\gname[j]$, etc.).
We take $\ptpset$, $\msgset$, and $\nodeset$ pairwise disjoint.
The participants of a choreography exchange messages to coordinate
with each other.
In the global view, this is modelled with
\emph{interactions}\footnote{ We depart from the usual notation
  $\gint[][@][][@]:\msg$ to a have a more lightweight syntax.  }
$\gint[{}]$, which represent the fact that participant $\ptp$ sends
message $\msg$ to participant $\ptp[B]$, which is expected to receive
$\msg$.
A \emph{global choreography} (g-choreography for short) is a term
$\aG$ derived by the following grammar (recursion is omitted for
simplicity as discussed in \cref{conc:sec})
\begin{equation}\label{eq:chor}
\aG \bnfdef \gempty \bnfmid \gint \bnfmid \gseq[{}] \bnfmid \gpar \bnfmid \gcho
\end{equation}
A g-choreography can be empty, a simple interaction, the sequential or
parallel composition of g-choreographies, or the choice between two
g-choreographies.
We implicitly assume $\ptp \neq \ptp[B]$ in interactions $\gint$.
In \eqref{eq:chor}, a \emph{control point} $\gname$ tags interaction,
choice, and parallel g-choreographies: we assume that in a
g-choreography $\aG$ any two control points occurring in different
positions are different, e.g., we cannot write
$\gpar[@][{\gint[{\gname[j]}]}][{\gint[@][{\ptp[C]}][{\msg[y]}][{\ptp[D]}]}]$.
Control points are a technical device (as we will see when defining
projections and semantics of g-choreographies) and they could be
avoided.%
\footnote{At the cost of adding technical complexity, one can
  automatically assign a unique identifier to such control points.}
Let $\gset$ be the set of g-choreographies and, for $\aG \in \gset$,
let $\enode{\aG}$ denote the set of control points in $\aG$.
Throughout the paper we may omit control points when immaterial, e.g.,
writing $\gcho[]$ instead of $\gcho$.
Finally, fix a function
$\nmerge[][] : \gset \to (\nodeset \to \nodeset)$ such that, for all
$\aG \in \gset$, $\nmerge[](\aG)$ (written $\nmerge[][\aG]$)
\begin{itemize}
\item is bijective when restricted to $\enode{\aG}$ and
\item for all $\gname \in \enode{\aG}$, $\nmerge[@][\aG] \not\in
  \enode{\aG}$.
\end{itemize}
As clear in \cref{sem:sec} (where we map g-choreographies on
hypergraphs), $\nmerge[]$ will be used to establish a bijective
relation between fork and merge control points corresponding to
choices (and, in \cref{hyper:sec}, for a bijective correspondence
between (control points of) complementary send/receive actions).
Finally, we take g-choreographies up to the structural congruence
relation induced by the following axioms:
\begin{itemize}
\item $\gcho[{}][\_][\_]$ and $\gpar[{}][\_][\_]$ form commutative
  monoids with respect to $\gempty$
\item $\gseq[{}][\_][\_]$ is associative, and $\gseq[{}][@][\gempty] = \aG$,
and $\gseq[{}][\gempty][\aG] = \aG$
\end{itemize}

The syntax in \eqref{eq:chor} captures the structure of a visual
language of directed acyclic graphs\footnote{Cycles are not considered
  for simplicity and can be easily added.} so that each
g-choreography $\aG$ can be represented as a rooted graph with a
single ``enter'' (``exit'') control point; that is $\aG$ has a distinguished
\emph{source} (resp. \emph{sink}) control point that can reach (resp. be
reached by) any other control point in $\aG$.
\begin{figure}[th]
  \centering
  $\begin{array}{c@{\qquad}c@{\qquad}c@{\qquad}c@{\qquad}c}
    % gempty
      \begin{tikzpicture}[node distance=1cm and 1cm, every node/.style={scale=.7,transform shape}]
        \node[src] at (0,0) (src) {}; \node[sink, below=of src] (sink)
        {}; \path[line] (src) -- (sink);
      \end{tikzpicture}
     &
     % gseq
     \begin{tikzpicture}[node distance=0.4cm and 0.4cm, every node/.style={scale=.7,transform shape}]
       \node[bblock] at (0,0) (g)
       {$\aG$}; \node[node, below=of
       g,label=left:$\gname$] (s1) {}; \node[bblock, below=of s1]
       (gp)
       {$\aG'$}; \path[line,dotted] (g) -- (s1); \path[line,dotted]
       (s1) -- (gp);
     \end{tikzpicture}
     &
     % gint
     \begin{tikzpicture}[node distance=0.4cm and 0.4cm, every node/.style={scale=.7,transform shape}]
       \node[src] at (0,2) (src) {};
       \node[block,label=right:$\gname$] at (0,1) (int)
       {$\gint[{}]$}; \node[sink] at (0,0) (sink) {};
       \path[line] (src) -- (int);
       \path[line] (int) -- (sink);
     \end{tikzpicture}
     &
     % gpar
     \begin{tikzpicture}[node distance=0.6cm and 0.6cm, every node/.style={scale=.7,transform shape}]
       \node[src] at (0,0) (src) {}; \node[agate,below=of
       src,label=below:$\gname$] (par) {}; \node[node, left=of par]
       (s1) {}; \node[node, right=of par] (s2) {}; \node[bblock,
       below=of s1] (g)
       {$\aG$}; \node[bblock, below=of s2] (gp)
       {$\aG'$}; \node[node, below=of g] (sg) {}; \node[node,
       below=of gp] (sgp) {}; \node[agate, right=of
       sg,label=above:$\nmerge$] (join) {}; \node[sink, below=of
       join] (sink) {}; \path[line] (src) -- (par); \path[line]
       (par) -- (s1); \path[line] (par) -- (s2); \path[line,dotted]
       (s1) -- (g); \path[line,dotted] (s2) -- (gp);
       \path[line,dotted] (g) -- (sg); \path[line,dotted] (gp) --
       (sgp); \path[line] (sg) -- (join); \path[line] (sgp) --
       (join); \path[line] (join) -- (sink);
     \end{tikzpicture}
     &
     % gcho
     \begin{tikzpicture}[node distance=0.6cm and 0.6cm, every node/.style={scale=.7,transform shape}]
       \node[src] at (0,0) (src) {}; \node[ogate,below=of
       src,label=below:$\gname$] (par) {}; \node[node, left=of par]
       (s1) {}; \node[node, right=of par] (s2) {}; \node[bblock,
       below=of s1] (g)
       {$\aG$}; \node[bblock, below=of s2] (gp)
       {$\aG'$}; \node[node, below=of g] (sg) {}; \node[node,
       below=of gp] (sgp) {}; \node[ogate, right=of
       sg,label=above:$\nmerge$] (join) {}; \node[sink, below=of
       join] (sink) {}; \path[line] (src) -- (par); \path[line]
       (par) -- (s1); \path[line] (par) -- (s2); \path[line,dotted]
       (s1) -- (g); \path[line,dotted] (s2) -- (gp);
       \path[line,dotted] (g) -- (sg); \path[line,dotted] (gp) --
       (sgp); \path[line,dotted] (sg) -- (join); \path[line,dotted]
       (sgp) -- (join); \path[line] (join) -- (sink);
     \end{tikzpicture}
     \\
     \text{empty graph}
     &
     \text{sequential}
     &
     \text{interaction}
     &
     \text{parallel}
     &
     \text{branching}
   \end{array}$
   \caption{Our graphs: $\gsource$ is the source node, $\gsink$ the
     sink one; other nodes are drawn as $\gvertex$}
   \label{fig:graphs}
\end{figure}
\Cref{fig:graphs} illustrates this; a dotted edge from/to a
$\gvertex$-control points single out the source/sink control point the
edge connects to.
For instance, in the graph for the sequential composition, the
top-most edge identifies $\aG$ sink node and the other edge identifies
the source node of $\aG'$; intuitively, $\gvertex$ is the control point
of the sequential composition of $\aG$ and $\aG'$ obtained by \lq\lq
coalescing\rq\rq\ the sink control point of $\aG$ with the source
control point of $\aG'$.
%, so the dotted edge from $\aG$ is the incoming edge of the sink control point of $\aG$; likewise, the dotted edge from $\gvertex$ is the unique outgoing edge of the source of $\aG'$.
%
In a graph $\aG \in \gset$, to each node $\gname$ of a branch/fork
corresponds the node $\nmerge[\gname][\aG]$ of its control point.
Labels will not be depicted when immaterial.
Our graphs resemble the global graphs of~\cite{dy12,lty15} the only
differences being that
\begin{itemize}
\item by construction, forking and branching control points $\gname$
  have a corresponding join and merge control point $\nmerge$;
\item there is a unique sink control point with a unique incoming edge (as
  in~\cite{dy12,lty15}, there is also a unique source control point with a
  unique outgoing edge).
\end{itemize}

As an example, consider the graph (where the control points of
interactions are omitted for readability)

\[\begin{tikzpicture}[node distance=.6cm and 1cm, every node/.style={scale=.7,transform shape}]
  \node[src] at (0,0) (src) {};
  \node[agate,below=of src,label=below:$\gname$] (par) {};
  \node[node, left=of par] (s1) {};
  \node[node, right=of par] (s2) {};
  \node[bblock, below=of s1] (g) {$\gint[]$};
  \node[bblock, below=of s2] (gp) {$\gint[][@][n]$};
  \node[node, below=of g] (sg) {};
  \node[node, below=of gp] (sgp) {};
  \node[agate, right=of sg,label=above:$\nmerge$] (join) {};
  \node[sink, below=of join] (sink) {};
  \path[line] (src) -- (par);
  \path[line] (par) -- (s1);
  \path[line] (par) -- (s2);
  \path[line] (s1) -- (g);
  \path[line] (s2) -- (gp);
  \path[line] (g) -- (sg);
  \path[line] (gp) -- (sgp);
  \path[line] (sg) -- (join);
  \path[line] (sgp) -- (join);
  \path[line] (join) -- (sink);
\end{tikzpicture}
\]
representing a choreography where $\ptp$ sends $\ptp[B]$ messages
$\msg$ and $\msg[n]$ in any order.

%%% Local Variables:
%%% mode: latex
%%% TeX-master: "main"
%%% End:

%  LocalWords:  choreographies monoid bijective hypergraphs

\section{Hypergraphs of events}\label{hyper:sec}
The semantics of a choice-free g-choreography $\aG \in \gset$ (i.e. a
choreography that does not contain $\gcho[{}][\_][\_]$ terms) is a
partial order, which represents the causal dependencies of the
communication actions specified by $\aG$.
Choices are a bit more tricky.
Intuitively, the semantics of $\gcho$ consists of two partial
orders, one representing the causal dependencies of the
communication actions of $\aG$ and the other of those of $\aG'$.
In the following, we will use hypergraphs as a compact representations
of sets of partial orders.

Actions \emph{happen on channels}, which we identify by the names of
the participants involved in the communication.
Formally, a channel is an element of the set
$\chset = \pset^2 \setminus \{(\p,\p) \st \p \in \pset\}$ and we
abbreviate $(\p,\q) \in \chset$ as $\achan$.
The set of \emph{events} $\eset$ (ranged over by $\ae$, $\ae'$,
$\ldots$) is defined by
%\rgnote{Here we also need $\nmerge$}
\[
\eset = \eset^! \cup \eset^? \cup \nodeset
\qquad\text{where}\qquad
\eset^! = \chanset \times  \{!\} \times \nodeset  \times \msgset
\qquad\text{and}\qquad
\eset^? = \chanset \times   \{?\} \times \nodeset \times \msgset
\]
Sets $\eset^!$ and $\eset^?$, the output and the input events,
respectively represent \emph{sending} and \emph{receiving} actions; we
shorten $(\achan,!,\gname,\msg)$ as $\aout$ and
$(\achan,?,\gname,\msg)$ as $\ain$.
The \emph{subject} of an action is
\[
\subject[\aout] = \p
\quad \text{($\p$ is the sender)}
\qquad\text{and}\qquad
\subject[\ain] = \p[B]
\quad
\text{($\p[B]$ is the receiver)}
\]
As will be clear later, events in $\nodeset$ represent
\quo{non-observable} actions, like (the execution of) a choice or a
merge; we take $\subject[\_]$ to be undefined on $\nodeset$.
We now continue by defining some auxiliary operations.

The \emph{communication action} of $\ae$ is
$\eact \ain = \ain[@][@][]$ and $\eact \aout = \aout[@][@][]$ and
undefined on $\nodeset$; we extend $\enode{}$ to events, so
$\enode \ae$ denotes the control point of an event $\ae$.
When considering sets of events $\aE \in \parts{\eset}$, we will
tacitly assume that any two events have different control points (that
is for all $\ae, \ae' \in \aE,\ \enode \ae \neq \enode{\ae'}$).
Also, we write $\ae \in \aG$ when there is an interaction $\gint$ in
$\aG$ such that $\ae \in \{\aout,\ain\}$, and accordingly $\aE
\subseteq \aG$ means that $\ae \in \aG$ for all $\ae \in \aE$.

A relation $\aR \subseteq \parts{\eset} \times \parts{\eset}$ on sets
of events is a directed hypergraph, that is a graph where nodes are
events and hyperarcs $\hopair$ relate sets of events, the source $\aE$
and the target $\aE'$.
(To avoid cumbersome parenthesis, singleton sets in hyperarcs are
shortened by their element, e.g., we write $\hopair[\ae][\aE]$
instead of $\hopair[\hyedge{\ae}][\aE]$.)
\begin{figure}[t]
  \centering
  \begin{subfigure}[b]{0.3\textwidth}
    \centering
  \begin{tikzpicture}[->,>=stealth',shorten >=1pt,auto,node distance=2cm, thick,main node/.style={font=\sffamily\bfseries},scale=.7,transform shape]

      \node[main node] (1) {$\aout[\p][\q][\gname_1][{\msg[x]}]$};
      \node[main node] (2) [below =0.7cm of 1] {$\ain[\p][\q][\gname_1][{\msg[x]}]$};
      \node[main node] (3) [below =0.7cm of 2] {$\aout[\q][\p][\gname_2][{\msg[y]}]$};
      \node[main node] (4) [below =0.7cm of 3] {$\ain[\q][\p][\gname_2][{\msg[y]}]$};

      \path[every node/.style={font=\sffamily\small}]
      (1) edge node {} (2)
      (2) edge node {} (3)
      (3) edge node {} (4)
      ;
    \end{tikzpicture}
    \caption{$\aR_\eqref{fig:hypergraphs:3}$}
    \label{fig:hypergraphs:3}
  \end{subfigure}
  ~
  \begin{subfigure}[b]{0.3\textwidth}
    \centering
    \begin{tikzpicture}[->,>=stealth',shorten >=1pt,auto,node distance=2cm,
      thick,main node/.style={
        font=\sffamily\bfseries},scale=.7,transform shape]

      \node[main node] (S) {$\gname_3$};
      \node[main node] (1) [below left of=S] {$\aout[\p][\q][\gname_1][{\msg[x]}]$};
      \node[main node] (2) [below of=1] {$\ain[\p][\q][\gname_1][{\msg[x]}]$};
      \node[main node] (3) [below right of=S] {$\aout[\p][\q][\gname_2][{\msg[y]}]$};
      \node[main node] (4) [below of=3] {$\ain[\p][\q][\gname_2][{\msg[y]}]$};
      \node[main node] (E) [below right of=2] {$\nmerge[\gname $_3$]$};

      \path[every node/.style={font=\sffamily\small}]
      (1) edge node {} (2)
      (3) edge node {} (4)

      (S) edge node {} (1)
      (S) edge node {} (3)

      (2) edge node {} (E)
      (4) edge node {} (E)
      ;
    \end{tikzpicture}
    \caption{$\aR_\eqref{fig:hypergraphs:1}$}
    \label{fig:hypergraphs:1}
  \end{subfigure}
  ~
  \begin{subfigure}[b]{0.3\textwidth}
    \centering
  \begin{tikzpicture}[->,>=stealth',shorten >=1pt,auto,node distance=2cm, thick,main node/.style={font=\sffamily\bfseries},scale=.7,transform shape]

      \node[main node] (S) {$\gname_3$};
      \node[main node] (f) [below =0.5cm of S] {};
      \node[main node] (1) [below left of=S] {$\aout[\p][\q][\gname_1][{\msg[x]}]$};
      \node[main node] (2) [below of=1] {$\ain[\p][\q][\gname_1][{\msg[x]}]$};
      \node[main node] (3) [below right of=S] {$\aout[\p][\q][\gname_2][{\msg[y]}]$};
      \node[main node] (4) [below of=3] {$\ain[\p][\q][\gname_2][{\msg[y]}]$};
      \node[main node] (E) [below right of=2] {$\nmerge[\gname $_3$]$};
      \node[main node] (j) [above =0.5cm of E] {};

      \path[every node/.style={font=\sffamily\small}]
      (1) edge node {} (2)
      (3) edge node {} (4)

      (S) edge[-] (f)
      (0,-0.8) edge (1)
      (0,-0.8) edge (3)

      (2) edge [-] (0,-4)
      (4) edge [-] (0,-4)

      (j) edge (E)
      ;
    \end{tikzpicture}
    \caption{$\aR_\eqref{fig:hypergraphs:2}$}
    \label{fig:hypergraphs:2}
  \end{subfigure}
  \caption{Some hypergraphs\label{fig:hypergraphs}}
\end{figure}
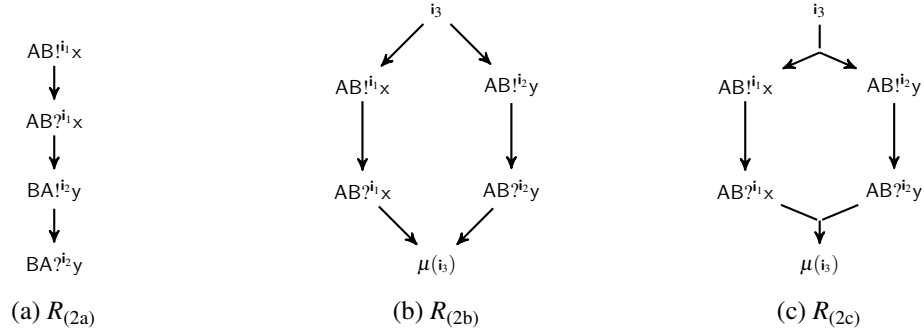
Examples of hypergraphs are depicted in \cref{fig:hypergraphs}; the
graphs $\aR_\eqref{fig:hypergraphs:3}$ and
$\aR_\eqref{fig:hypergraphs:1}$ contain only simple arcs, while the
graph $\aR_\eqref{fig:hypergraphs:2}$ contains two hyperarcs:
$\hopair[\gname_3][{\{ \aout[\p][\q][\gname_1][{\msg[x]}],
  \aout[\p][\q][\gname_2][{\msg[y]}] \}}]$
and
$\hopair[{\{ \ain[\p][\q][\gname_1][{\msg[x]}],
  \ain[\p][\q][\gname_2][{\msg[y]}] \}}][{\nmerge[{$\gname_3$}]}]$.
Intuitively, $\aR_\eqref{fig:hypergraphs:3}$ establishes a total
causal order from the top-most to the bottom-most event;
$\aR_\eqref{fig:hypergraphs:1}$ represents a choice at control point
$\gname_3$ between the left and the right branch; finally,
$\aR_\eqref{fig:hypergraphs:2}$ represents the parallel execution of
two threads at the control point $\gname_3$; note that the edge
$\hopair[{\gname_3}][{\{ \aout[\p][\q][\gname_1][{\msg[x]}],
  \aout[\p][\q][\gname_2][{\msg[y]}] \}}]$
of $\aR_\eqref{fig:hypergraphs:2}$ relates the event $\gname_3$ to
both $\aout[\p][\q][\gname_1][{\msg[x]}]$ and
$\aout[\p][\q][\gname_2][{\msg[y]}]$.

Let $\efst{1}{},\efst{2}{} : 2^\eset \times 2^\eset \to 2^\eset$ be
the maps projecting a relation on its components, that is:
$\efst{1}{\hopair[\aE][\aE']} = \aE$ and $\efst{2}{\hopair[\aE][\aE']}
= \aE'$.
Given $\aR,\aR' \subseteq \parts{\eset} \times \parts{\eset}$, define
the hypergraphs $\aR \circ \aR'$ and $\rtrs[\aR]$ respectively as
\[
  \aR \circ \aR' =
  \{\hopair \st \exists \hopair[@][\aE_1] \in \aR, \hopair[\aE_2][\aE'] \in \aR' \qst \aE_1 \cap \aE_2 \neq \emptyset\}
  \qquad \text{and} \qquad
  \rtrs[\aR] = \bigcup_n \underbrace{\aR \circ \cdots \circ \aR}_{n\text{-times}}
\]
Basically, $\rtrs[\aR]$ is the reflexo-transitive closure of $\aR$
with respect to the composition relation $\circ$.
In \cref{fig:causal} we give a simple example of how operation
$\_ \circ \_$ composes hyperedges (thick arrows) according to the
underlying causal relations (thin arrows); edges $\hopair$ and
$\hopair[\ae_i'][\aE'']$ are composed to form the edge
$\hopair[@][\aE'']$, which relates each event in $\aE$ to all those in
$\aE''$.

We define the maximal and minimal elements of $\aR$ respectively as
\[
  \rmax = \{ \ae \in \eset \st \not\exists \hopair \in \aR \wedge
              \ae \in \aE\}
  \qquad \text{and} \qquad
  \rmin = \{ \ae \in \eset \st \not\exists \hopair \in \aR \wedge
              \ae \in \aE' \}
\]
For instance, $\aR_\eqref{fig:hypergraphs:1}$ and
$\aR_\eqref{fig:hypergraphs:2}$ in \cref{fig:hypergraphs}
respectively have
$\rmin[\aR_\eqref{fig:hypergraphs:1}] =
\rmin[\aR_\eqref{fig:hypergraphs:2}] = \{\gname_3\} $
and
$ \rmax[\aR_\eqref{fig:hypergraphs:1}] =
\rmax[\aR_\eqref{fig:hypergraphs:2}] = \{\nmerge[{\gname $_3$}]\}$,
while the minimal and maximal elements of
$\aR_\eqref{fig:hypergraphs:3}$ are
$\aout[\p][\q][\gname_1][{\msg[x]}]$ and
$\ain[\q][\p][\gname_2][{\msg[y]}]$ respectively.
We also need to define the (hyperedges involving) \quo{last} and the
\quo{first} communication actions in $\aR$.
\[
  \rMAX = \{ \hopair \in \aR \st \aE' \cap \nodeset = \emptyset  \ \land\
  \forall \hopair[\aE'][\aE''] \in \rtrs[\aR] \qst \aE'' \subseteq \nodeset \}
  \qquad\text{and}\qquad
  \rMIN  =  \big(\rMAX[{(\aR^{-1})}]\big)^{-1}
  % \rMIN & = & \{ \hopair \in \aR \st \hopair[\aE'][\aE] \in \rMAX[{(\aR^{-1})}] \}
\]
For instance, the \quo{first} and the \quo{last} communication actions
of $\aR_\eqref{fig:hypergraphs:3}$ in \cref{fig:hypergraphs} are $\{
\hopair [{\aout[\p][\q][\gname_1][{\msg[x]}]}]
[{\ain[\p][\q][\gname_1][{\msg[x]}]}] \}$ and $\{ \hopair
[{\aout[\q][\p][\gname_2][{\msg[y]}]}]
[{\ain[\q][\p][\gname_2][{\msg[y]}]}] \}$ respectively, while
$\aR_\eqref{fig:hypergraphs:1}$ and $\aR_\eqref{fig:hypergraphs:2}$
have the same \quo{first} and the \quo{last} communication actions ($
\rMIN[\aR_\eqref{fig:hypergraphs:1}] =
\rMIN[\aR_\eqref{fig:hypergraphs:2}] = \{ \hopair
[{\aout[\p][\q][\gname_1][{\msg[x]}]}]
[{\ain[\p][\q][\gname_1][{\msg[x]}]}], \hopair
[{\aout[\p][\q][\gname_2][{\msg[y]}]}]
[{\ain[\p][\q][\gname_2][{\msg[y]}]}] \} =
\rMAX[\aR_\eqref{fig:hypergraphs:1}] =
\rMAX[\aR_\eqref{fig:hypergraphs:2}] $).

We can now define $\rseq[\aR][\aR']$, the \emph{sequential}
composition of relations $\aR$ and $\aR'$ on $\eset$ as follows:
\begin{align*}
  \rseq[\aR][\aR']\ = & \quad \aR \cup \aR'
 \cup \big\{\hopair[\ae][\ae'] \in \big(\parts{\eset \setminus \nodeset}\big)^2 \st
           \exists \hopair[\aE_1][\aE_2] \in \rMAX,
           \hopair[\aE'_1][{\aE'_2}] \in \rMIN[\aR'] \qst
           \\  & 
           \hspace{3.5cm}
           \ae \in (\aE_1 \cup \aE_2) \setminus \nodeset
           \ \land\
           \ae' \in (\aE'_1 \cup \aE'_2) \setminus \nodeset
           \ \land\
           \subject = \subject[\ae'] \big\}
\end{align*}
The sequential composition of two hypergraphs $\aR$ and $\aR'$
preserves the causal dependencies of its constituents, namely those in
$\aR \cup \aR'$.
Additionally, dependencies are established between every event in $\rMAX$
and every event in $\rMIN[\aR']$ that have the same subject.
\cref{fig:seq} depicts the sequential compositions of two
hypergraphs, say $\aR$ and $\aR'$.
The former hypergraph corresponds to the interaction $\gint$, while
the second ranges over the interactions
\[
\gint[\gname'][@][{\msg[y]}][{\p[C]}]
\qquad
\gint[\gname'][{\p[B]}][{\msg[y]}][{\p[C]}]
\qquad
\gint[\gname'][{\p[C]}][{\msg[y]}][{\p[B]}]
\qquad
\gint[\gname'][@][{\msg[y]}][@]
\qquad
\gint[\gname'][{\p[C]}][{\msg[y]}][{\p[D]}]
\]
with the events at control point $\gname$ belonging to $\aR$ and those
at control point $\gname[i]'$ belonging to $\aR'$; also, simple arrows
represent the dependencies induced by the subjects and dotted arrows
represent dependencies induced by the sequential composition (the
meaning of stroken arrows will be explained in \cref{sem:sec});
basically a causal relation is induced whenever a participant
performing a (last) communication of $\aR$ also starts a communication
in $\aR'$.

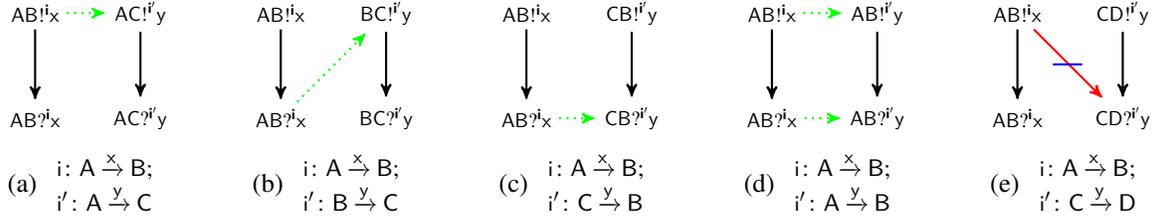
\begin{figure}[t]
  \centering
  \begin{subfigure}[b]{0.15\textwidth}
    \centering
    \begin{tikzpicture}[->,>=stealth',shorten >=1pt,auto,node distance=2cm,
      thick,main node/.style={font=\sffamily\bfseries},scale=.7,transform shape]
      \node[main node] (1) {$\aout[\p][\q][\gname][{\msg[x]}]$};
      \node[main node] (2) [below of=1] {$\ain[\p][\q][\gname][{\msg[x]}]$};
      \node[main node] (3) [right of=1] {$\aout[\p][{\p[C]}][\gname'][{\msg[y]}]$};
      \node[main node] (4) [below of=3] {$\ain[\p][{\p[C]}][\gname'][{\msg[y]}]$};
      \path[every node/.style={font=\sffamily\small}]
      (1) edge node {} (2)
      (1) edge[green,dotted] node {} (3)
      (3) edge node {} (4)
      ;
    \end{tikzpicture}
    \caption{\footnotesize
      $
      \begin{array}[l]{l}
        \gseq[][{\gint[\gname][\p][{\msg[x]}][\q]}][]
        \\
        \gint[\gname'][\p][{\msg[y]}][{\p[C]}]
      \end{array}
      $
    }
    \label{fig:seq:1}
  \end{subfigure}
  \qquad
  \begin{subfigure}[b]{0.15\textwidth}
    \centering
    \begin{tikzpicture}[->,>=stealth',shorten >=1pt,auto,node distance=2cm,
      thick,main node/.style={font=\sffamily\bfseries},scale=.7,transform shape]
      \node[main node] (1) {$\aout[\p][\q][\gname][{\msg[x]}]$};
      \node[main node] (2) [below of=1] {$\ain[\p][\q][\gname][{\msg[x]}]$};
      \node[main node] (3) [right of=1] {$\aout[\q][{\p[C]}][\gname'][{\msg[y]}]$};
      \node[main node] (4) [below of=3] {$\ain[\q][{\p[C]}][\gname'][{\msg[y]}]$};
      \path[every node/.style={font=\sffamily\small}]
      (1) edge node {} (2)
      (2) edge[green,dotted] node {} (3)
      (3) edge node {} (4)
      ;
    \end{tikzpicture}
    \caption{\footnotesize
      $
      \begin{array}[l]{l}
        \gseq[][{\gint[\gname][\p][{\msg[x]}][\q]}][]
        \\
        \gint[\gname'][\q][{\msg[y]}][{\p[C]}]
      \end{array}
      $
    }
    \label{fig:seq:2}
  \end{subfigure}
  \qquad
  \begin{subfigure}[b]{0.15\textwidth}
    \centering
    \begin{tikzpicture}[->,>=stealth',shorten >=1pt,auto,node distance=2cm,
      thick,main node/.style={font=\sffamily\bfseries},scale=.7,transform shape]
      \node[main node] (1) {$\aout[\p][\q][\gname][{\msg[x]}]$};
      \node[main node] (2) [below of=1] {$\ain[\p][\q][\gname][{\msg[x]}]$};
      \node[main node] (3) [right of=1] {$\aout[{\p[C]}][\q][\gname'][{\msg[y]}]$};
      \node[main node] (4) [below of=3] {$\ain[{\p[C]}][\q][\gname'][{\msg[y]}]$};
      \path[every node/.style={font=\sffamily\small}]
      (1) edge node {} (2)
      (3) edge node {} (4)
      (2)  edge[green,dotted] node {} (4)
      ;
    \end{tikzpicture}
    \caption{\footnotesize
      $
      \begin{array}[l]{l}
        \gseq[][{\gint[\gname][\p][{\msg[x]}][\q]}][]
        \\
        \gint[\gname'][{\p[C]}][{\msg[y]}][\q]
      \end{array}
      $
    }
    \label{fig:seq:3}
  \end{subfigure}
  \qquad
  \begin{subfigure}[b]{0.15\textwidth}
    \centering
    \begin{tikzpicture}[->,>=stealth',shorten >=1pt,auto,node distance=2cm,
      thick,main node/.style={font=\sffamily\bfseries},scale=.7,transform shape]
      \node[main node] (1) {$\aout[\p][\q][\gname][{\msg[x]}]$};
      \node[main node] (2) [below of=1] {$\ain[\p][\q][\gname][{\msg[x]}]$};
      \node[main node] (3) [right of=1] {$\aout[\p][\q][\gname'][{\msg[y]}]$};
      \node[main node] (4) [below of=3] {$\ain[\p][\q][\gname'][{\msg[y]}]$};
      \path[every node/.style={font=\sffamily\small}]
      (1) edge node {} (2)
      (1) edge[green,dotted] node {} (3)
      (3) edge node {} (4)
      (2)  edge[green,dotted] node {} (4)
      ;
    \end{tikzpicture}
    \caption{\footnotesize
      $
      \begin{array}[l]{l}
        \gseq[][{\gint[\gname][\p][{\msg[x]}][\q]}][]
        \\
        \gint[\gname'][\p][{\msg[y]}][\q]
      \end{array}
      $
    }
    \label{fig:seq:4}
  \end{subfigure}
  \qquad
  \begin{subfigure}[b]{0.15\textwidth}
    \centering
    \begin{tikzpicture}[->,>=stealth',shorten >=1pt,auto,node distance=2cm,
      thick,main node/.style={font=\sffamily\bfseries},scale=.7,transform shape]
      \node[main node] (1) {$\aout[\p][\q][\gname][{\msg[x]}]$};
      \node[main node] (2) [below of=1] {$\ain[\p][\q][\gname][{\msg[x]}]$};
      \node[main node] (3) [right of=1] {$\aout[{\p[C]}][{\p[D]}][\gname'][{\msg[y]}]$};
      \node[main node] (4) [below of=3] {$\ain[{\p[C]}][{\p[D]}][\gname'][{\msg[y]}]$};
      \path[every node/.style={font=\sffamily\small}]
      (1) edge node {} (2)
      (3) edge node {} (4)
      (1) edge[red,strike thru arrow] node {} (4)
      ;
    \end{tikzpicture}
    \caption{\footnotesize
      $
      \begin{array}[l]{l}
        \gseq[][{\gint[\gname][\p][{\msg[x]}][\q]}][]
        \\
        \gint[\gname'][{\p[C]}][{\msg[y]}][{\p[D]}]
      \end{array}
      $
    }
    \label{fig:seq:error}
  \end{subfigure}
  \caption{Examples of sequential composition}
  \label{fig:seq}
\end{figure}

We now define the concept of \quo{common} part of two hypergraphs $\aR$
and $\aR'$ with respect to a participant $\p$.
For this we need to introduce the \emph{happens-before} relation
\[
\widehat \aR = \{ \opair \in \eset \times \eset \st \exists \hopair
\in \aR \qst \ae \in \aE \text{ and } \ae' \in \aE'\} \subseteq \eset
\times \eset
\]
induced by a relation $\aR$ ($\opair \in \widehat \aR$ when $\ae$
precedes $\ae'$ in $\aR$, namely $\widehat \aR$ are the causal
dependencies among the events in $\aR$).
\cref{fig:causal} yields an intuitive representation of how causal
relations follow composition: the events in $\aE$ cause all the events
in $\aE''$ due to the dependency of the event $\ae_i'$ from the events
in $\aE$ and the fact that $\ae_i$ causes all events in $\aE''$.
%\end{document}
\begin{figure}[t]
  \centering
    \begin{tikzpicture}[scale=.65,transform shape]
    \node (x) {};
    \node[right=of x] (E) {$\aE$};
    \node[right=of E] (eq) {$ = \{ $};
    \node[right=of eq] (e1) {$\ae_1$};
    \node[right=of e1] (d1) {$\cdots $};
    \node[right=of d1] (eh) {$\ae_h$};
    \node[right=of eh] (c) {$\}$};
    \node[below=of E] (db) {};
    \node[below=of c] (de) {};
%    \node[below=of E] {$\hopair$};
    %
    \node[below=of db] (E') {$\aE'$};
    \node[right=of E'] (eq') {$ = \{ $};
    \node[right=of eq'] (e1') {$\ae_1'$};
    \node[right=of e1'] (c') {$\cdots $};
    \node[right=of c'] (eh') {$\ae_i'$};
    \node[below=of de] (c') {$\}$};
    \node[below=of E'] (db') {};
    \node[below=of c'] (de') {};
    \path (E) edge[line,line width=2pt,color=black!20] node[sloped,above,color=black] {$\hopair$} (E');
    \node[below=of db'] (E2) {$\aE''$};
    \node[right=of E2] (eq'') {$= \{$};
    \node[right=of eq''] (e1'') {$\ae_1''$};
    \node[right=of e1''] (d1'') {$\cdots $};
    \node[right=of d1''] (eh'') {$\ae_j''$};
    \node[right=of eh''] (c'') {$\}$};
    \path[line] (e1) -- (e1');
    \path[line,color=black!20] (e1) -- (e1');
    \path[line] (e1) -- (eh');
    \path[line,color=green] (e1) -- (eh');
    \path[line] (eh) -- (e1');
    \path[line,color=black!20] (eh) -- (e1');
    \path (E') edge[line,line width=2pt,color=red,dashed] node[sloped,above,color=black] {$\hopair[\ae_i'][\aE'']$} (E2);
    \path[line] (eh) -- (eh');
    \path (eh) edge[line,color=green]  (eh');
%    \node[below =of E'] {$\hopair[\ae_h'][\aE'']$};
    \path[line, color=red] (eh') -- (e1'');
    \path[line, color=red] (eh') -- (eh'');
    \path (E) edge[line, line width=3pt, bend right, bend angle=20] node[sloped,  yshift=.5cm, rotate=180] {$\hopair[@][\aE'']$} (E2);
    %
    % \only<3->{\node[right=of c] (min) {$\rmin[]$}};
    % \only<4->{\node[right=of min] (fst) {$\rMIN[]$}};
    % \only<3->{\node[right=of c''] (max) {$\rmax[]$}};
    % \only<4->{\node[right=of max] (lst) {$\rMAX[]$}};
  \end{tikzpicture}
  \caption{Happens-before\label{fig:causal}}
\end{figure}
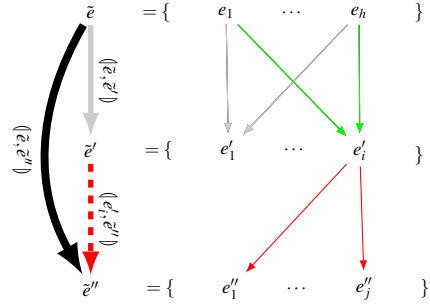

A set of events $\aE'$ in $\aR'$  \emph{$\p$-reflects} a set of
  events $\aE$ in $\aR$ if there is a bijection
$\common : \aE \to \aE'$ such that:
\begin{itemize}
\item $\forall \ae \in \aE \qst \subject = \subject[\common(\ae)] = \p
  \ \land\ \eact{\ae} = \eact{\common(\ae)}$ and
\item
  $\forall \ae' \in \aE \ \forall \opair \in \widehat \aR \qst
  \subject[\ae] = \p \implies \big(\ae \in \aE \ \land\
  \opair[\common(\ae)][\common(\ae')] \in \widehat{\aR'}\big)$ and
\item $\forall \ae' \in \common(\aE) \ \forall \opair \in \widehat \aR' \qst \subject[\ae] = \p
  \implies \big( \ae \in \common(\aE) \ \land\ \opair[\common^{-1}(\ae)][\common^{-1}(\ae')] \in
  \widehat{\aR}\big)$.
\end{itemize}
The notion of reflections is new;  an
intuitive explanation is given in \cref{fig:reflectivity}.
\begin{figure}
  \centering
  \begin{tabular}{cl}
    \begin{minipage}{.4\linewidth}
      \includegraphics[scale=.35]{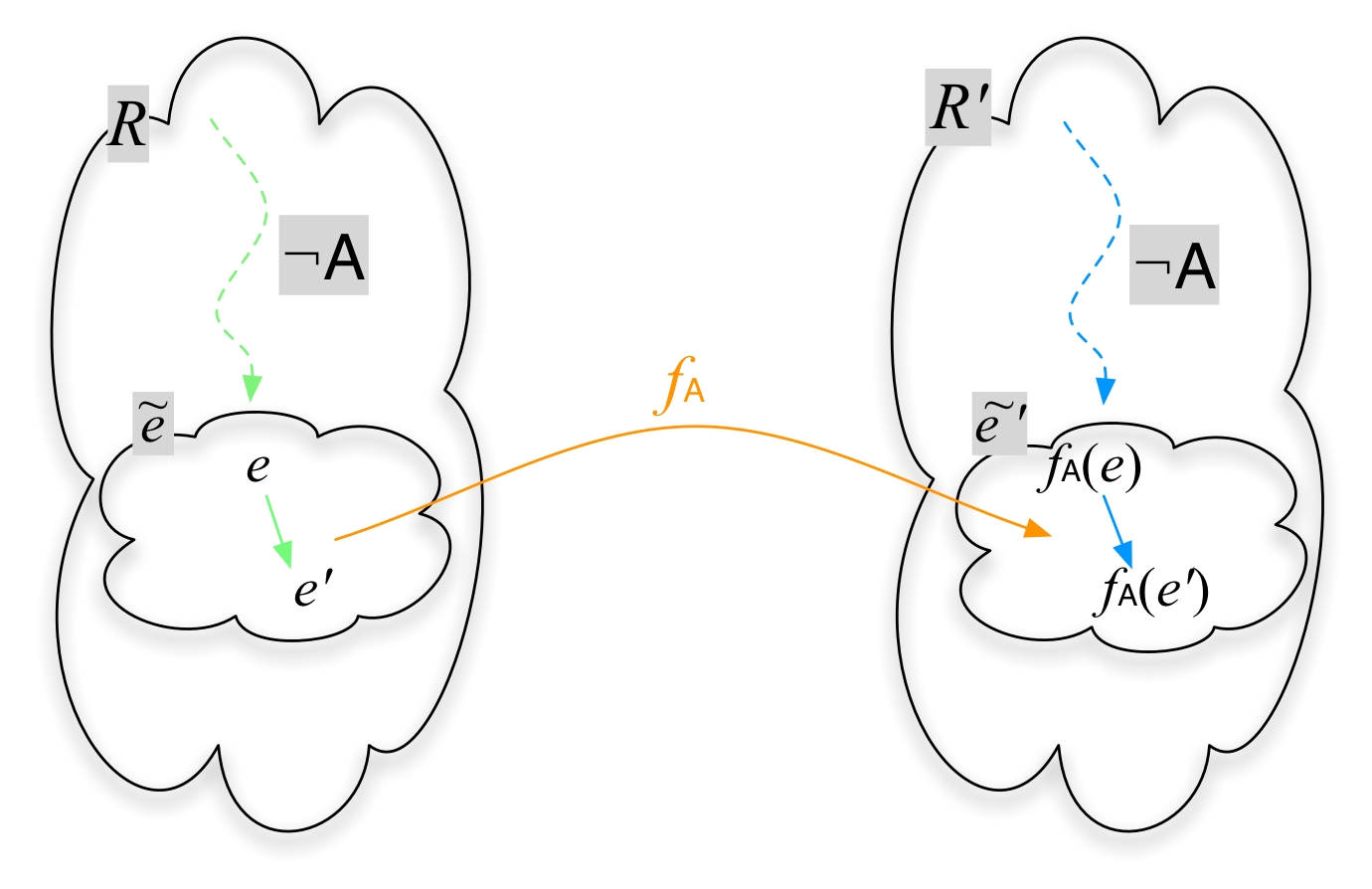}
    \end{minipage}
    &
    \begin{minipage}{.5\linewidth}
      The causal relations of $R$ and $R'$ have to be thought of as
      the ones of two branches of a distributed choice.
      All the events of $\aE \subseteq \aR$ have subject $\p$, the
      selector of the choice.
      Likewise for $\aE' \subseteq \aR'$.
      
      The bijection $\common$ preserves both actions and causality
      relation in $\aE$.
      Moreover, $\aE$ have to be such that any event with subject $\p$
      causing an event of $\aE$ is also a member of $\aE$, and
      similarly for $\aE'$.
    \end{minipage}
  \end{tabular}
  \caption{Reflectivity\label{fig:reflectivity}}
\end{figure}
Reflectivity will allow us to define \emph{active} and \emph{passive}
participants in a choice.

%%% Local Variables:
%%% mode: latex
%%% TeX-master: "main"
%%% End:

%  LocalWords:  hypergraphs hypergraph hyperarcs singlelton reflexo hyperedges
%  LocalWords:  stroken bijection

\section{Semantics of Choreographies }\label{sem:sec}
The semantics of g-choreography is the partial map
$\gsem[\_][{\nmerge[][]}] : \gset \to \parts{(\parts{\eset}
  \times \parts{\eset})}$ defined\footnote{We assume $\nmerge[][]$ to
  be understood and simply write $\gsem[\_]$.} as:
\begin{align*}
  \gsem[\gempty] & = \emptyset  %\label{seq:0}
  \\
  \gsem[\gint]  & = \{ \hopair[\aout][\ain]\}
  \\
  \gsem[\gpar] & = \gsem[\aG] \cup \gsem[\aG']
  \\
  \gsem[{\gseq[]}] & = \begin{cases}
    \rseq[{\gsem[\aG]}][{\gsem[\aG']}]%\aR
    & \mbox{if }
    \rtrs[{\big(\rseq[{\gsem[\aG]}][{\gsem[\aG']}]\big)}] \supseteq \efst{1}{\rMAX[{\gsem}]} \times \efst{2}{\rMIN[{\gsem[\aG']}]}
    \\
    \rbot & \mbox{otherwise}
  \end{cases}
  \\
  \gsem[\gcho] & =  %\label{seq:cho}
  \begin{cases}
    \gsem[\aG] \cup \gsem[\aG'] \cup \aR
    &
    \mbox{if } \aR = 
    \{
    \hopair[\gname][{\rmin[{\gsem[\aG]}]}],
    \hopair[\gname][{\rmin[{\gsem[\aG']}]}],
    \hopair[{\rmax[{\gsem[\aG]}]}][{\gname[{$\nmerge[@][]$}]}],
    \hopair[{\rmax[{\gsem[\aG']}]}][{\gname[{$\nmerge[@][]$}]}]
    \}
    \\
    &
    \text{ and }
    \wb[\aG][\aG']
    \\
    \rbot & \mbox{otherwise}
  \end{cases}
\end{align*}
The semantics of the the empty g-choreography $\gempty$ and of
interaction $\gint$ are straightforward; for the latter, the send part
$\aout$ of the interaction must precede its receive $\ain$ part.

For the parallel composition $\gpar$ we just take the union of the
dependencies of $\aG$ and $\aG'$, thus allowing the arbitrary
interleaving of those events.

The semantics of sequential composition $\gseq$ establishes
happens-before relations as computed by
$\rseq[{\gsem[\aG]}][{\gsem[\aG']}]$ provided that they cover the
dependencies between the last communication actions of $\aG$ with the
first actions of $\aG'$.
This condition ensures the soundness of the composition; when it does
not hold, then there is a participant $\ptp$ in $\aG'$ that cannot
ascertain if all the events of $\aG$ did happen before $\ptp$ could
start.
All examples in \cref{fig:seq} are sound, barred the one in
\cref{fig:seq:error}, where the stroken edge depicts the missing
dependency that is not guaranteed by the hypergraph.

The semantics of a choice $\gcho$ is defined provided that the
\emph{well-branched} condition $\wb[\aG][\aG']$ holds on $\aG$ and
$\aG'$, that is when ($i$) there is at most one \emph{active}
participant and ($ii$) all the other participants are
\emph{passive}.
In a moment, after some auxiliary definitions, we define active and
passive participants.
Intuitively, the notions of active and passive participant single out
respectively participants $\p$ that do not make an internal choice,
namely it is not $\p$ selecting whether to execute $\aG$ or $\aG'$ and
those participants instead that (internally) select which branch to
execute.
Besides the dependencies induced by $\aG$ and $\aG'$, $\gsem[\gcho]$
contain those making $\gname$ (the control point of the branch)
precede all minimal events of $\aG$ and $\aG'$; similarly, the maximal
events of $\aG$ and $\aG'$ have to precede the conclusion of the
choice (marked by the control point $\nmerge$).
Notice that no additional dependency is required.
In fact, during one instance of the g-choreography either the actions
of the first branch or the actions of the second one will be
performed.

\paragraph{Auxiliary definitions}
The relation $\gord$ is the \emph{happens-before relation induced by
  $\aG \in \gset$} defined as $\gord = \widehat{(\rtrs[\gsem])}$ if
$\gsem$ is defined, and $\gord = \emptyset$ otherwise.
Notice that $\gord$ is a partial order on the events of $\aG$.
For $\p \in \ptpset$, the \emph{$\p$-only} part of a set of events
$\aE \in \parts{\eset}$ is the set $\eproj$ where the actions of $\aE$
not having subject $\p$ are replaced with the control point of the
action; formally
\begin{eqnarray*}
  \eproj & = &
  \{\ae \in \aE \mid \subject[\ae] = \p \ \lor\ \ae \in \nodeset\}
  \\  &  &  \cup
  \{\enode \ae \mid \ae \in \aE \cap \eset^! \ \land\ \subject  \neq \p\}
  \cup % \\ & \cup  &
  \{\nmerge[{\enode \ae}][] \mid \ae \in \aE \cap \eset^? \ \land\ \subject  \neq \p\}
\end{eqnarray*}
Accordingly, the \emph{$\p$-only} part of a hypergraphs $\aR$ is
defined as
$\hproj = \left \{ \hopair[\eproj][{\eproj[{\aE'}]}] \st \hopair
  [\aE][\aE'] \in \aR \right \}$.
Notice that we use $\enode \ae$ and $\nmerge[{\enode \ae}][]$ for
outputs and inputs respectively, so that different events not
belonging to $\p$ remain distinguished.

Given a participant $\p \in \ptpset$, two g-choreographies
$\aG,\aG' \in \gset$, and two sets of events $\aE \subseteq \aG$ and
$\aE' \subseteq \aG'$ the \emph{$\p$-branching pair of $\gcho[]$ with
  respect to $\aE$ and $\aE'$} (written $\rrdiv$) is
\[
\rrdiv = (\aE_1, \aE_2)
\qquad \text{where}\qquad
\aE_1  =  \bigcup{\efst{1}{\rMIN[{(\hproj[\gsem])}]} \setminus \aE}
\quad\text{and}\quad
\aE_2  = \bigcup{\efst{1}{\rMIN[{(\hproj[{\gsem[{\aG'}]}])}]} \setminus \aE'}
\]
provided that $\aE'$ $\p$-reflects $\aE$ (otherwise $\rrdiv$ is undefined).
Intuitively, the behaviour of $\p$ in the two branches $\aG$ and
$\aG'$ can be the same up to the point of branching $\rrdiv$.
The \emph{$\p$-reflectivity} is used to identify such common behaviour
(i.e.  all events in $\aE$ and $\aE'$) and to ignore it when checking
the behaviour of $\p$ in the branches.
In fact, by taking the $\p$-only parts of these hypergraphs and
selecting their fist interactions (that is the $\p$-branching pair
$\aE_1$, $\aE_2$) we identify when the behaviour of $\p$ in $\aG$
starts to be different with respect to behaviour in $\aG'$.

\paragraph{Active and passive roles}
The intersection of sets of events $\tilde e \sqcap \tilde e'$
disregards control points:
$\tilde e \sqcap \tilde e' = \{act(e): e \in \tilde e\} \cap
\{act(e'): e' \in \tilde e'\}$.
A participant $\p \in \ptpset$ is \emph{passive} in $\gcho[]$ with
respect to $\aE$ and $\aE'$ if, assuming $(\aE_1, \aE_2) = \rrdiv$,
the following hold
\[
\begin{array}[c]{c@{\qquad\qquad}c}
  \aE_1 \sqcap \{\ae \in \aG' \st \not
  \exists \ae' \in \aE_2 \qst \ae \gord[\aG'] \ae' \} = \emptyset
  &
  \aE_1 \cup \aE_2 \subseteq \eset^?
  \\
  \aE_2 \sqcap \{\ae \in \aG \st \not \exists \ae' \in \aE_1 \qst \ae \gord \ae' \} = \emptyset
  &
  \aE_1 = \emptyset \iff \aE_2 = \emptyset
\end{array}
\]
Thus, the behaviour of $\p$ in $\aG$ and $\aG'$ must be the same up to
a point where she receives either of two different messages, each one
identifying which branch had been selected.
Clearly, $\p$ cannot perform outputs at the points of branching.
We say that a participant $\ptp$ is \emph{passive} in $\gcho[]$ if 
such $\aE$ and $\aE'$ exist.

A participant $\p \in \ptpset$ is \emph{active} in $\gcho[][@][@]$
with respect to $\aE$ and $\aE'$ if, assuming
$(\aE_1, \aE_2) = \rrdiv$,
\[
\begin{array}[c]{c@{\qquad\qquad}c@{\qquad\qquad}c@{\qquad\qquad}c}
  \aE_1 \cup \aE_2 \subseteq \eset^! &
  \aE_1 \sqcap \aE_2 = \emptyset &
  \aE_1 \neq \emptyset & \aE_2 \neq \emptyset
\end{array}
\]
Thus, the behaviour of $\p$ in $\aG$ and $\aG'$ must be the same up to
the point where she informs the other participants, by sending
different messages, which branch she choses.
We say that a participant $\ptp$ is \emph{active} in $\gcho[]$ if 
such $\aE$ and $\aE'$ exist.
Interestingly, if one takes the empty reflection in the determination
of active and passive roles, the definition above yield exactly the
same notions used e.g., in~\cite{honda16jacm,bmt14,cdyp16}.

\paragraph{Some examples}
When it exists, the active participant is the selector of the choice.
Unlike its corresponding notions in the rest of the literature,
well-branchedness does not require the selector to exist.
For instance, the choreography
\[\begin{tikzpicture}[node distance=.6cm and 1cm, scale = .7, every node/.style={transform shape}]
  \node[src] at (0,0) (src) {};
  \node[ogate,below=of src,label=below:$\gname$] (par) {};
  \node[node, left=of par] (s1) {};
  \node[node, right=of par] (s2) {};
  \node[bblock, below=of s1] (g) {$\gint[]$};
  \node[bblock, below=of s2] (gp) {$\gint[]$};
  \node[node, below=of g] (sg) {};
  \node[node, below=of gp] (sgp) {};
  \node[ogate, right=of sg,label=above:$\nmerge$] (join) {};
  \node[sink, below=of join] (sink) {};
  \path[line] (src) -- (par);
  \path[line] (par) -- (s1);
  \path[line] (par) -- (s2);
  \path[line] (s1) -- (g);
  \path[line] (s2) -- (gp);
  \path[line] (g) -- (sg);
  \path[line] (gp) -- (sgp);
  \path[line] (sg) -- (join);
  \path[line] (sgp) -- (join);
  \path[line] (join) -- (sink);
\end{tikzpicture}
\qquad = \qquad
% \gcho[@][{\gint[{\gname[j]}]}][{\gint[{\gname[j']}]}]
\gcho[@][{\gint[]}][{\gint[]}]
\]
is well-branched even if it has no active participant.
Another example (usually discharged in the literature by imposing
syntactic constraints) is
$\gcho[]
[{\gseq[][{\gint[{\gname[i]}]}][{{\gint[][{\ptp[B]}][{\msg[x]}][C]}}]}]
[{\gseq[][{\gint[{\gname[j]}]}][{\gint[][{\ptp[B]}][{\msg[y]}][C]}]}]$;
here the problem is that the two branches have the same first
interactions.
However, using reflection on the $\hopair[\aout][\ain]$ and
$\hopair[{\aout[@][@][{\gname[j]}]}][{\ain[@][@][{\gname[j]}]}]$, our
framework establishes that $\ptp[B]$ is active, and both $\ptp$ and
$\ptp[C]$ are passive, making the choice well-branched.
We are not aware of any other framework where the cases above are
considered valid choreographies.

The hypergraphs in \cref{fig:hypergraphs:1} and
\cref{fig:choice} are respectively the semantics of
the g-choreographies
\begin{align}
  \aG_\eqref{fig:hypergraphs:1} & = 
  \gcho[\gname_3]
  [{
    \gint[\gname_1][\p][{\msg[x]}][\q]
  }]
  [{
    \gint[\gname_2][\p][{\msg[y]}][\q]
  }]
  \\
  \aG_\eqref{fig:choice:2} & =
  \gcho[\gname_3]
  [{
    \gint[\gname_1][\p][{\msg[x]}][\q]
  }]
  [{
    \gint[\gname_2][\p][{\msg[y]}][{\ptp[C]}]
  }]
  \\
  \aG_\eqref{fig:choice:3} & = 
  \gcho[\gname_5]
  [{ \left (
      \gseq[]
      [{
        \gint[\gname_1][\p][{\msg[x]}][\q]
      }]
      [{
        \gint[\gname_2][\q][{\msg[y]}][{\ptp[C]}]
      }]
    \right )
  }]
  [{ \left (
      \gseq[]
      [{
        \gint[\gname_3][\p][{\msg[z]}][{\ptp[C]}]
      }]
      [{
        \gint[\gname_4][{\ptp[C]}][{\msg[w]}][\q]
      }]
    \right )
  }]
\end{align}
\textbf{\cref{fig:hypergraphs:1}} the choice is well-branched;
participant $\q$ is passive (receiving either
$\ain[\p][\q][][{\msg[x]}]$ or $\ain[\p][\q][][{\msg[y]}]$ in the
point of branching) and participant $\p$ is active (sending either
$\aout[\p][\q][][{\msg[x]}]$ or $\aout[\p][\q][][{\msg[y]}]$ in the
point of branching).
\\
\textbf{\cref{fig:choice:2}} the choice is not well-branched;
participant $\p$ is active (sending either
$\aout[\p][\q][][{\msg[x]}]$ or $\aout[\p][{\ptp[C]}][][{\msg[y]}]$ in
the point of branching), however, $\q$ (and $\ptp[C]$) is neither
passive nor active (in one branch the events of branching is
$\ain[\p][\q][][{\msg[x]}]$ while for the other branch it is empty).
\\
\textbf{\cref{fig:choice:3}} the choice is well-branched; $\p$ is
active (sending either $\aout[\p][\q][][{\msg[x]}]$ or
$\aout[\p][{\ptp[C]}][][{\msg[z]}]$ in the point of branching), $\q$
is passive (it receives either $\ain[\p][\q][][{\msg[x]}]$ or
$\ain[{\ptp[C]}][\q][][{\msg[w]}]$ in the events of branching), and
$\ptp[C]$ is passive (it receives either
$\ain[\q][{\ptp[C]}][][{\msg[y]}]$ or
$\ain[\p][{\ptp[C]}][][{\msg[z]}]$ in the branching events).
\\
\textbf{\cref{fig:choice:4}} the choice is well-branched; $\p$ is
active (it has the same behaviour in the branches $\gname_3$ and
$\gname_6$, so its branching events are
$\aout[\p][{\ptp[C]}][][{\msg[z]}]$ and
$\aout[\p][{\ptp[C]}][][{\msg[w]}]$), $\q$ is passive (having the same
behaviour in the branches $\gname_3$ and $\gname_6$ and empty sets of
branching), and $\ptp[C]$ is passive (its branching events are the
inputs $\ain[\p][{\ptp[C]}][][{\msg[z]}]$
$\ain[\p][{\ptp[C]}][][{\msg[w]}]$).
%
% \rgnote{Suggest to remove this example, since it is no more well-branched and we obtain some space}.
% \textbf{\cref{fig:choice:5}} the choice is well-branched;
% $\p$ is passive (it has the same behaviour in the branches
% $\aout[\p][{\ptp[C]}][][{\msg[x]}]$), $\q$ is passive (it has the same
% behaviour in the branches $\aout[\q][{\ptp[C]}][][{\msg[y]}]$), and
% $\ptp[C]$ is passive (it receives either
% $\ain[\p][{\ptp[C]}][][{\msg[x]}]$ or
% $\ain[\q][{\ptp[C]}][][{\msg[y]}]$ in the branches events).
% \rgnote{Notice that this does not work if the messages in the second
% branch are $z$ and $w$ instead of $x$ and $y$}.
% % 
% Participants $\q$ and $\ptp[C]$ are neither passive or active (in one
% branch the branching events is $\ain[\p][\q][][{\msg[x]}]$ while
% for the other branch it is empty).

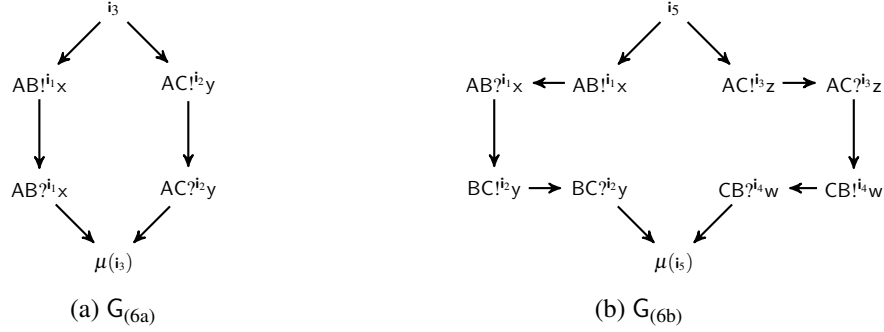
\begin{figure}[ht]
  \centering
  % \begin{subfigure}[b]{0.3\textwidth}
  %   \centering
  %   \begin{tikzpicture}[->,>=stealth',shorten >=1pt,auto,node distance=2cm,
  %     thick,main node/.style={
  %     font=\sffamily\bfseries}]

  %     \node[main node] (S) {$\gname_3$};
  %     \node[main node] (1) [below left of=S] {$\aout[\p][\q][\gname_1][{\msg[x]}]$};
  %     \node[main node] (2) [below of=1] {$\ain[\p][\q][\gname_1][{\msg[x]}]$};
  %     \node[main node] (3) [below right of=S] {$\aout[\p][\q][\gname_2][{\msg[y]}]$};
  %     \node[main node] (4) [below of=3] {$\ain[\p][\q][\gname_2][{\msg[y]}]$};
  %     \node[main node] (E) [below right of=2] {$\nmerge[\gname $_3$]$};

  %     \path[every node/.style={font=\sffamily\small}]
  %     (2) edge node {} (1)
  %     (4) edge node {} (3)

  %     (1) edge node {} (S)
  %     (3) edge node {} (S)

  %     (E) edge node {} (2)
  %     (E) edge node {} (4)
  %     ;
  %   \end{tikzpicture}
  %   \caption{$\aG_\eqref{fig:choice:1}$}
  %   \label{fig:choice:1}
  % \end{subfigure}
  % ~
  \begin{subfigure}[t]{0.3\textwidth}
    \centering
    \begin{tikzpicture}[->,>=stealth',shorten >=1pt,auto,node distance=2cm,
      thick,main node/.style={
        font=\sffamily\bfseries},scale=.7,transform shape]

      \node[main node] (S) {$\gname_3$};
      \node[main node] (1) [below left of=S] {$\aout[\p][\q][\gname_1][{\msg[x]}]$};
      \node[main node] (2) [below of=1] {$\ain[\p][\q][\gname_1][{\msg[x]}]$};
      \node[main node] (3) [below right of=S] {$\aout[\p][{\ptp[C]}][\gname_2][{\msg[y]}]$};
      \node[main node] (4) [below of=3] {$\ain[\p][{\ptp[C]}][\gname_2][{\msg[y]}]$};
      \node[main node] (E) [below right of=2] {$\nmerge[\gname $_3$]$};

      \path[every node/.style={font=\sffamily\small}]
      (1) edge node {} (2)
      (3) edge node {} (4)

      (S) edge node {} (1)
      (S) edge node {} (3)

      (2) edge node {} (E)
      (4) edge node {} (E)
      ;
    \end{tikzpicture}
    \caption{$\aG_\eqref{fig:choice:2}$}
    \label{fig:choice:2}
  \end{subfigure}
  \hspace{2cm}
  \begin{subfigure}[t]{0.3\textwidth}
    \centering
    \begin{tikzpicture}[->,>=stealth',shorten >=1pt,auto,node distance=2cm,
      thick,main node/.style={
        font=\sffamily\bfseries},scale=.7,transform shape]

      \node[main node] (S) {$\gname_5$};

      \node[main node] (1) [below left of=S] {$\aout[\p][\q][\gname_1][{\msg[x]}]$};
      \node[main node] (2) [left of=1] {$\ain[\p][\q][\gname_1][{\msg[x]}]$};
      \node[main node] (1N1) [below of=2] {$\aout[\q][{\ptp[C]}][\gname_2][{\msg[y]}]$};
      \node[main node] (1N2) [right of=1N1] {$\ain[\q][{\ptp[C]}][\gname_2][{\msg[y]}]$};

      \node[main node] (3) [below right of=S] {$\aout[\p][{\ptp[C]}][\gname_3][{\msg[z]}]$};
      \node[main node] (4) [right of=3] {$\ain[\p][{\ptp[C]}][\gname_3][{\msg[z]}]$};

      \node[main node] (2N1) [below of=4] {$\aout[{\ptp[C]}][\q][\gname_4][{\msg[w]}]$};
      \node[main node] (2N2) [left of=2N1] {$\ain[{\ptp[C]}][\q][\gname_4][{\msg[w]}]$};

      \node[main node] (E) [below right of=1N2] {$\nmerge[\gname $_5$]$};

      \path[every node/.style={font=\sffamily\small}]
      (1) edge node {} (2)
      (2) edge node {} (1N1)
      (1N1) edge node {} (1N2)

      (3) edge node {} (4)
      (4) edge node {} (2N1)
      (2N1) edge node {} (2N2)

      (S) edge node {} (1)
      (S) edge node {} (3)

      (1N2) edge node {} (E)
      (2N2) edge node {} (E)
      ;
    \end{tikzpicture}
    \caption{$\aG_\eqref{fig:choice:3}$}
    \label{fig:choice:3}
  \end{subfigure}
  ~
  \caption{Some examples\label{fig:choice}}
\end{figure}

\begin{figure}[t]
  \centering
  \begin{tikzpicture}[->,>=stealth',shorten >=1pt,auto,node distance=2cm,
    thick,main node/.style={
      font=\sffamily\bfseries},scale=.7,transform shape]
    \node[main node] (S0) {$\gname_7$};
    \node[main node] (S) [below left = 0.1cm and 2cm  of S0] {$\gname_3$};
    \node[main node] (1) [below left of=S] {$\aout[{\ptp[A]}][{\ptp[B]}][\gname_1][{\msg[x]}]$};
    \node[main node] (2) [below of=1] {$\ain[{\ptp[A]}][{\ptp[B]}][\gname_1][{\msg[x]}]$};
    \node[main node] (3) [below right of=S] {$\aout[{\ptp[A]}][{\ptp[B]}][\gname_2][{\msg[y]}]$};
    \node[main node] (4) [below of=3] {$\ain[{\ptp[A]}][{\ptp[B]}][\gname_2][{\msg[y]}]$};
    \node[main node] (E) [below right of=2] {$\nmerge[\gname $_3$]$};
    \node[main node] (AS) [right of=S, node distance=5cm] {$\gname_6$};
    \node[main node] (A1) [below left of=AS] {$\aout[{\ptp[A]}][{\ptp[B]}][\gname_4][{\msg[x]}]$};
    \node[main node] (A2) [below of=A1] {$\ain[{\ptp[A]}][{\ptp[B]}][\gname_4][{\msg[x]}]$};
    \node[main node] (A3) [below right of=AS] {$\aout[{\ptp[A]}][{\ptp[B]}][\gname_5][{\msg[y]}]$};
    \node[main node] (A4) [below of=A3] {$\ain[{\ptp[A]}][{\ptp[B]}][\gname_5][{\msg[y]}]$};
    \node[main node] (AE) [below right of=A2] {$\nmerge[\gname $_6$]$};
    \node[main node] (E0) [below right = 0.1cm and 2cm  of E] {$\nmerge[\gname $_7$]$};
    \node[main node] (BS) [below of=E0] {$\gname_{10}$};
    \node[main node] (B1) [below left of=BS] {$\aout[{\ptp[A]}][{\ptp[C]}][\gname_8][{\msg[z]}]$};
    \node[main node] (B2) [below of=B1] {$\ain[{\ptp[A]}][{\ptp[C]}][\gname_8][{\msg[z]}]$};
    \node[main node] (B3) [below right of=BS] {$\aout[{\ptp[A]}][{\ptp[C]}][\gname_9][{\msg[w]}]$};
    \node[main node] (B4) [below of=B3] {$\ain[{\ptp[A]}][{\ptp[C]}][\gname_9][{\msg[w]}]$};
    \node[main node] (BE) [below right of=B2] {$\nmerge[\gname $_{10}$]$};
    \path[every node/.style={font=\sffamily\small}]
    (1)    edge node {}    (2) 
    (3)    edge node {}    (4) 
    
    (S)    edge node {}    (1) 
    (S)    edge node {}    (3) 
    
    (2)    edge node {}    (E) 
    (4)    edge node {}    (E) 
    
    (A1)   edge node {}   (A2)
    (A3)   edge node {}   (A4)
    
    (AS)   edge node {}   (A1)
    (AS)   edge node {}   (A3)
    
    (A2)   edge node {}   (AE)
    (A4)   edge node {}   (AE)
    
    (B1)   edge node {}   (B2)
    (B3)   edge node {}   (B4)
    
    (BS)   edge node {}   (B1)
    (BS)   edge node {}   (B3)
    
    (B2)   edge node {}   (BE)
    (B4)   edge node {}   (BE)
    
    (S0)   edge node {}    (S) 
    (S0)   edge node {}   (AS)
    (E)    edge node {}   (E0)
    (AE)   edge node {}   (E0)
    
    (E0)   edge node {}   (BS)
    ;
  \end{tikzpicture}
  \caption{$
    \gseq[]
    [{
      \gcho[\gname_7]
      [{
        \gcho[\gname_3]
        [{
          \gint[\gname_1][{\ptp[A]}][{\msg[x]}][{\ptp[B]}]
        }]
        [{
          \gint[\gname_2][{\ptp[A]}][{\msg[y]}][{\ptp[B]}]
        }]
      }]
      [{
        \gcho[\gname_6]
        [{
          \gint[\gname_4][{\ptp[A]}][{\msg[x]}][{\ptp[B]}]
        }]
        [{
          \gint[\gname_5][{\ptp[A]}][{\msg[y]}][{\ptp[B]}]
        }]
      }]
    }]
    [{
      \gcho[\gname_{10}]
      [{
        \gint[\gname_8][{\ptp[A]}][{\msg[z]}][{\ptp[C]}]
      }]
      [{
        \gint[\gname_9][{\ptp[A]}][{\msg[w]}][{\ptp[C]}]
      }]
    }]
    $}
  \label{fig:choice:4}
\end{figure}

\section{Languages of Choreographies }\label{lang:sec}
The abstract semantics of a g-choreography is a hypergraph, which
represents the set of partial orders among the events of the
g-choreography.
A more concrete semantics can be given by considering the
\emph{language} of a g-choreography.
%
% Informally, such languages are computed by ($i$) computing the
% hypergraph representing the semantics of the g-choreagraphy, ($ii$)
% iterating all possible decisions for the g-choreography choices, ($iii$)
% extracting from the hypergraph the partial order corresponding to the
% decision made and ($iv$) generating the words that satisfy the partial
% order.
%
Informally, the language of a g-choreography $\aG \in \gset$ consists
of the sequences of words made of the communication actions of the
events in $\aG$ that preserve the causal relations of $\gsem$,
provided that $\gsem$ is defined.

Given a g-choreography $\aG$, let
$\ich = \gsem \cap (\parts \nodeset \times \parts \eset)$ be the set
of \emph{choice hyperedges} of $\aG$ (that is those hyperedges in
$\aG$ whose source represents choices) and define the outgoing
hyperedges of $\gname \in \nodeset$ in $\aG$ as
$\ichedges = \ich \cap (\{\hyedge \gname\} \times \parts \eset)$.
A map $\actch: \ich \to \parts {\eset}$ is a \emph{resolution of
  $\aG$} if $\actch(\gname) \in \ich(\gname)$
for every $\gname \in \nodeset$.
Intuitively, a resolution fixes a branch for every choice in a
g-choreography $\aG$ and therefore it induces a preorder of the events
compatible with $\aG$ and the resolution.

The preorder corresponding to a resolution is computed by
$\rOnActchClean$. This hypergraph is obtained by ($i$) removing 
every hyperedge not
chosen by the resolution and
($ii$) removing every dead event (i.e. events that are not reachable from the
initial events after removing the non-selected hyperedges):
\[
\rOnActchClean \quad = \quad (\mathit{trim}\big(
\gsem \setminus \bigcup_{\gname \in \ich}(\ich(\gname) \setminus \actch(\gname))
,\rmin[\gsem])\big)
\]
where $trim(\aR, \aE)$ is the function that removes every node in the
hypergraph $\aR$ that is not reachable from $\aE$ and
$\aR \setminus \aE = \left \{ \hopair[{\aE_1 \setminus \aE}] [{\aE_2
    \setminus \aE}] \st \hopair [\aE_1][\aE_2] \in \aR \right \}$.
\eMnote{nota che il risultato potrebbe essere eg $(\emptyset, ...)$}

%The function $\rOnActch$ yields the preorder where  is removed:
%\[
%\rOnActch \quad = \quad \gsem \setminus \bigcup_{\gname \in \ich}(\ich(\gname) \setminus \actch(\gname))
%\]
%Finally, we use $\rAllEvents[\rOnActchClean]$ to identify the list of all events occurring
%in a preorder $\rOnActchClean$. \eMnote{???}

\newcommand{\lalph}{\mathcal{A}}

Let $\lalph = \eset^! \cup \eset^?$.
% and, for $\aW \in \lalph^{*}$, $\size \aW$ be the length of $\aW$.%
The \emph{language} of $\aG \in \gset$ is
\[
\rlang[\aG] = \{\eact \aW \st \aW \in \lalph^{*} \mbox{ and }
\exists \mbox{ a resolution } \actch \mbox{ of } \aG \qst \psi(\aW,
\actch) \}
\]
where, $\psi(\aW, \actch)$ holds iff for all $i \neq j$ between
$1$ and the length of $\aW$ we have that
\begin{enumerate}
\item \label{lang:4} $\widx[\aW][i] \neq \widx[\aW][j]$,
  where $\widx[\aW][i]$ stands for the $i$-th symbol in $\aW$
\item \label{lang:2} $\widx[\aW][i], \widx[\aW][j] \in \rOnActchClean$
\item \label{lang:1} if
  $\widx[\aW][i] \gord[\rOnActchClean] \widx[\aW][j]$ then $i<j$
\item \label{lang:3} for every $\ae$, if
  $ \ae \gord[\rOnActchClean] \widx[\aW][i]$ then there exists $h<i$
  such that $\widx[\aW][h] = \ae$
\end{enumerate}
\cref{lang:4,lang:2} state that events in the word are not repeated
and that the word is made only of events present in the preorder,
i.e. the word cannot mix events belonging to two different branches.
\cref{lang:1} states that words preserve the causal relations of events.
\cref{lang:3} requires that all the predecessors of an event in the
word must precede the event in the word.
Notice that $\rlang[\aG]$ is prefix-closed.

% \rgnote{
% There is a problem with this definition of language. The choreography
% $\gint$ has the following semantics: $\gord[{\gint}] = \{(\aout, \ain)\}$.
% The problem is that the biggest prefix closed language representing the
% constraints above is infinite. So we need at least require that each word of 
% the language contains at most one occurrence of the alphabet. However, in this 
% case the language is 
% }
% \[
% \left \{
%  \begin{array}{l}
%     \epsilon\\
%     \aout\\
%     \aout;\ain\\
%     \ain
%  \end{array}
% \right \}
% \]

% \rgnote{
% However, the last word would not be in the language, since the event $\ain$ 
% can not be in a trace if there is not the corresponding output
% }

%%% Local Variables:
%%% mode: latex
%%% TeX-master: "main"
%%% End:

%  LocalWords:  hypergraph choreagraphy hyperedges preorder hyperedge iff

% \section{Communicating Machines}\label{cfms:sec}
%\input{cfsm}

\section{Projecting on Communicating Machines}\label{proj:sec}%\label{cfms:sec}
\renewcommand{\vec}[1]{\tilde{#1}}

As in~\cite{lty15,dy12}, we adopt \emph{communicating finite state
  machines} (CFSM) as local artefacts.
We borrow the definition of CFSMs in~\cite{bz83}, with slight
adaptation to our context.
A CFSM is a finite transition system given by a tuple
$\aCM = (\aQ,\aQzero,\aTrs)$ where
\begin{itemize}
\item $\aQ$ is a finite set of {\em states} with $\aQzero \in \aQ$ the
  \emph{initial} state, and
\item $\aTrs\ \subseteq \ \aQ \times \eact{\lalph} \times \aQ$ is a set of
  \emph{transitions}; we write $q \trans{\ae} {q'}$ for
  $(q,\ae,q') \in
  \aTrs$.
\end{itemize}

A CFSM $(\aQ,\aQzero,\aTrs)$ is \emph{$\ptp$-local}  if for
every $q \trans{\ae} {q'} \in \aTrs$ holds $\subject[\ae] = \ptp$.
Given a \emph{$\ptp$-local} CFSM $\aCM_{\ptp} = (\aQ_{\ptp}, \aQzero[\ptp],
\aTrs_{\ptp})$ for each  $\ptp \in \ptpset$, the tuple $\aCS =
(\aCM_{\ptp})_{\ptp \in \ptpset}$ is a   \emph{communicating system}.

\newcommand{\osred}{\rightarrow}                % one-step reduction relation
\newcommand{\TRANS}[1]{\xrightarrow{#1}}
\newcommand{\R}{\osred}
\newcommand{\msred}{\osred^\ast}     % multi-step reduction
\newcommand{\RR}{\msred}
\newcommand{\ASET}[1]{\{ {#1} \}}
\newcommand{\SGV}{N}
\newcommand{\RSN}[2]{\RS_#1(#2)}
\newcommand{\RSB}[1]{\RSN{1}{#1}}
\newcommand{\TSdelta}{\hat{\delta}}

\newcommand{\abuffer}{b}

The semantics of communicating systems is defined in terms of 
\emph{transition systems}, which keep track of
the state of each machine and the content of each buffer.
  Let $\aCS = (\aCM_{\ptp})_{\ptp \in \ptpset}$ be a \emph{communicating
    system}.
  A \emph{configuration} of $\aCS$ is a pair $\aConf = \csconf q \abuffer$ where
  $\vec q = (q_{\ptp})_{\ptp \in \ptpset}$ with $q_{\ptp} \in \aQ_{\ptp}$ and where
  $\vec{\abuffer}=(\abuffer_{\ptp \ptp[B]})_{\ptp \ptp[B] \in\chset}$ with
  $\abuffer_{\ptp \ptp[B]}\in 
  \msgset^\ast$;
  $q_{\ptp}$ keeps track of the state of the machine $\ptp$ and
  $\abuffer_{\ptp \ptp[B]}$ is the buffer that keeps track of the messages
  delivered 
  from $\ptp$ to   $\ptp[B]$. 
  The \emph{initial} configuration $\aConf_0$ is the one where $q_{\ptp}$ is the
  initial state of the corresponding CFSM and all buffers are empty.

  A configuration $\aConf'= \csconf {q'} {\abuffer'} $ is {\em reachable} from 
  another configuration $\aConf = \csconf {q} {\abuffer}$ by \emph{firing
    transition $\ae$}, written $\aConf \TRANSS{\ae} \aConf'$
  if there is $\msg \in \msgset$ such that either (1) or (2) below
  hold:
  \begin{center}
    \begin{tabular}{l@{\hspace{.5cm}}r}
      \begin{minipage}{.45\linewidth}\small
	1. 
        $\ae = \aout[@][@][][@]$ and $q_{\ptp} \trans{\ae} {q'_{\ptp}} \in \aTrs_{\ptp}$
        and 
        \begin{itemize}
        \item[a.] $q_{\ptp[C]}' = q_{\ptp[C]}$ for all ${\ptp[C]} \neq \ptp$ 
        \item[b.]  and $\abuffer_{\ptp[A]\ptp[B]}' = \abuffer_{\ptp[A]\ptp[B]}.\msg$
        \item[c.] and 
          $\abuffer_{\ptp[A']\ptp[B']}'=\abuffer_{\ptp[A']\ptp[B']}$ for all $(\ptp[A'],\ptp[B']) \neq (\ptp[A],\ptp[B])$
        \end{itemize}
      \end{minipage}
      &
        \begin{minipage}{.45\linewidth}\small
          2.
        $\ae = \ain[@][@][][@]$ and $q_{\ptp} \trans{\ae} {q'_{\ptp}} \in \aTrs_{\ptp}$
        and 
        \begin{itemize}
        \item[a.] $q_{\ptp[C]}' = q_{\ptp[C]}$ for all ${\ptp[C]} \neq \ptp[B]$ 
        \item[b.]  and $\abuffer_{\ptp[A]\ptp[B]} = \msg.\abuffer'_{\ptp[A]\ptp[B]}$
        \item[c.] and 
          $\abuffer_{\ptp[A']\ptp[B']}'=\abuffer_{\ptp[A']\ptp[B']}$ for all $(\ptp[A'],\ptp[B']) \neq (\ptp[A],\ptp[B])$
        \end{itemize}
        \end{minipage}
    \end{tabular}
  \end{center}
  Condition (1) puts $\msg$ on channel $\ptp\ptp[B]$, while (2) gets
  $\msg$ from channel $\ptp\ptp[B]$.

A configuration  $\aConf= \csconf {q} {\abuffer}$ is \emph{stable} if all buffers are
empty: $\vec \abuffer = \vec \NUL$.
A configuration  $\aConf= \csconf {q} {\abuffer}$ is a \emph{deadlock} if 
$\aConf \not \TRANSS{} $ and 
\begin{itemize}
  \item there exists a $\ptp \in \ptpset$ such that $q_{\ptp}
    \trans{\ain[@][@][][@]} {q'_{\ptp}} \in \aTrs_{\ptp}$

  \item or $\vec \abuffer \neq \vec \NUL$
\end{itemize}

The language of a communicating system $\aCS$ is the biggest prefix 
closed set $\rlang[\aCS] \in \eact{\lalph}^{\star}$
such that for each
$\ae_0 \dots \ae_{n-1} \in \rlang[\aCS]$,  
$\aConf_0 \TRANSS{\ae_0}  \dots
\TRANSS{\ae_{n-1}} \aConf_{n}$.

Given two CFSMs $\aM = (Q,q_0,\tset)$ and $\aM' = (Q',q_0',\tset')$,
write $\aM \ccup \aM'$ for the machine
$(Q \cup Q', q_0, \tset \cup \tset')$ provided that $q_0 = q_0'$;
also, $\aM \ccap \aM'$ denotes $Q \cap Q'$.
The product of $\aM$ and $\aM'$ is defined as usual as
$\aM \cprod \aM'  =  (Q \times Q', (q_0,q_0'), \tset'')$
where $\big((q_1,q'_1), \ae, (q_2,q'_2)\big) \in \tset''$ if, and only if,
\[
    \big((q_1, \ae, q_2) \in \tset \text{ and } q'_1 = q'_2\big)
    \qquad \text{or} \qquad
    \big((q'_1, \ae, q'_2) \in \tset' \text{ and } q_1 = q_2\big)
\]
We also use $min(M)$ to denote the CFSM obtained by minimising $M$
(using e.g., the classical partition refinement algorithm) when
interpreting them as finite automata.

Let $\aG$ be a g-choreography, the function $\gproj$ yields the
projection (in the form of a CFSM) of the choreography over the
participant $\ptp$ using $q_0$ and $q_e$ as initial and sink states
respectively.
The projection is defined as follow: 
\[
\geproj = 
\begin{cases}
    \tikz{
    \node at (0,0) (i) {};
    \node[shape=circle,draw,inner sep=1pt] at (0.7,0) (s) {$\cinit$};
    \node at (1.3,0) (f) {};
    \draw[->] (i) -- (s);
    \draw[->] (s) -- (f);
  }
  & \text{if } \aG = \gempty \text{ and } \cinit = \cfinal
  \\
    \tikz{
    \node at (0,0) (i) {};
    \node[shape=circle,draw,inner sep=1pt] at (0.7,0) (s) {$\cinit$};
    \node at (1.3,0) (f) {};
    \draw[->] (i) -- (s);
    \draw[->] (s) -- (f);
  }
  & \text{if } \aG = \gint[@][{\ptp[B]}][@][{\ptp[C]}]  \text{ and } \cinit = \cfinal
  \\
  \ctr{{\cinit}}{{\cfinal}}{\aout[\ptp][{\ptp[B]}][][m]} & \text{if } \aG = \gint \text{ and } \cinit \neq \cfinal
  \\
  \ctr{\cinit}{\cfinal}{\ain[{\ptp[B]}][{\ptp[A]}][][m]} & \text{if } \aG =
  \gint[@][{\ptp[B]}][@][{\ptp[A]}]
  \text{ and } \cinit \neq \cfinal
  \\
  \geproj[{\aG_1}][@][@][\cfinal'] \ccup \geproj[\aG_2][@][\cfinal'][\cfinal] &
  \text{if } \aG = \gseq[@][\aG_1][\aG_2] \text{ and }  \geproj[\aG_1][@][@][\cfinal'] \ccap \geproj[\aG_2][@][\cfinal'][\cfinal] = \{\cfinal'\}
  \\
  \geproj[{\aG_1}] \ccup \geproj[\aG_2] & \text{if } \aG = \gcho[@][\aG_1][\aG_2] \text{ and } \geproj[\aG_1] \ccap \geproj[\aG_2][@] = \{\cinit,\cfinal\}
  \\
  \geproj[{\aG_1}][@][{\overline \cinit}][{\overline \cfinal}] \times 
  \geproj[{\aG_2}][@][{\underline \cinit}][{\underline \cfinal}] &
  \text{if } \aG = \gpar[@][\aG_1][\aG_2], \ 
  \geproj[{\aG_1}][@][{\overline \cinit}][{\overline \cfinal}] \ccap
    \geproj[{\aG_2}][@][{\underline \cinit}][{\underline \cfinal}] = \emptyset,
\  \cinit = (\overline \cinit, \underline \cinit)
    \\
&
\ \ \ \ 
 \text{ and }
  \cfinal = (\overline \cfinal, \underline \cfinal)
\end{cases}
\]

\begin{figure}[ht]
  \centering
  \begin{subfigure}[b]{0.20\textwidth}
    \centering
    \begin{tikzpicture}[->,>=stealth',shorten >=1pt,auto,node distance=1.7cm,
      thick,main node/.style={
        font=\scriptsize\sffamily\bfseries}]

      \node[main node] (S) {$\overline {q_{0 \ptp}}$};
      \node[main node] (1) [below of=S] {$\overline {q_{e \ptp}}$};
      \node[main node] (2) [right of=S] {$\overline {q_{0 \ptp[B]}}$};
      \node[main node] (3) [below of=2] {$\overline {q_{e \ptp[B]}}$};

      \path[every node/.style={font=\sffamily\scriptsize}]
      (S) edge node {$\aout[@][@][][x]$} (1)
      (2) edge node {$\ain[@][@][][x]$} (3)
      ;
    \end{tikzpicture}
    \caption{$\gint[][@][x][@]$}
    \label{fig:proj:1}
  \end{subfigure}
~
  \begin{subfigure}[b]{0.20\textwidth}
    \centering
    \begin{tikzpicture}[->,>=stealth',shorten >=1pt,auto,node distance=1.7cm,
      thick,main node/.style={
        font=\scriptsize\sffamily\bfseries}]

      \node[main node] (S) {$\underline {q_{0 \ptp}}$};
      \node[main node] (1) [below of=S] {$\underline {q_{e \ptp}}$};
      \node[main node] (2) [right of=S] {$\underline {q_{0 \ptp[B]}}$};
      \node[main node] (3) [below of=2] {$\underline {q_{e \ptp[B]}}$};

      \path[every node/.style={font=\sffamily\scriptsize}]
      (S) edge node {$\aout[@][@][][y]$} (1)
      (2) edge node {$\ain[@][@][][y]$} (3)
      ;
    \end{tikzpicture}
    \caption{$\gint[][@][y][@]$}
    \label{fig:proj:2}
  \end{subfigure}
~
  \begin{subfigure}[b]{0.5\textwidth}
    \centering
    \begin{tikzpicture}[->,>=stealth',shorten >=1pt,auto,node distance=1.7cm,
      thick,main node/.style={
        font=\scriptsize\sffamily\bfseries}]

      \node[main node] (PS)  {$(\overline {q_{0 \ptp}}, \underline {q_{0 \ptp}})$};
      \node[main node] (1) [below left of = PS] {$(\overline {q_{1 \ptp}}, \underline {q_{0 \ptp}})$};
      \node[main node] (2) [below right of = PS] {$(\overline {q_{0 \ptp}}, \underline {q_{1 \ptp}})$};
      \node[main node] (E) [below right of = 1] {$(\overline {q_{1 \ptp}}, \underline {q_{1 \ptp}})$};

      \node[main node] (BPS) [right =3cm of  PS]  {$(\overline {q_{0 \ptp[B]}}, \underline {q_{0 \ptp[B]}})$};
      \node[main node] (B1)  [below left of = BPS] {$(\overline {q_{1 \ptp[B]}}, \underline {q_{0 \ptp[B]}})$};
      \node[main node] (B2)  [below right of = BPS] {$(\overline {q_{0 \ptp[B]}}, \underline {q_{1 \ptp[B]}})$};
      \node[main node] (BE)  [below right of = B1] {$(\overline {q_{1 \ptp[B]}}, \underline {q_{1 \ptp[B]}})$};

      \path[every node/.style={font=\scriptsize\sffamily}]
      (PS) edge node [xshift=-1cm,yshift=.5cm] {$\aout[@][@][][x]$} (1)
      (PS) edge node {$\aout[@][@][][y]$} (2)
      (1) edge node [xshift=-1cm,yshift=-.5cm] {$\aout[@][@][][y]$} (E)
      (2) edge node {$\aout[@][@][][x]$} (E)

      (BPS) edge node  [xshift=-1cm,yshift=.5cm] {$\ain[@][@][][x]$} (B1)
      (BPS) edge node {$\ain[@][@][][y]$} (B2)
      (B1) edge node  [xshift=-1cm,yshift=-.5cm] {$\ain[@][@][][y]$} (BE)
      (B2) edge node {$\ain[@][@][][x]$} (BE)
      ;
    \end{tikzpicture}
    \caption{$\gpar[][{\gint[][@][x][@]}][{\gint[][@][y][B]}]$}
    \label{fig:proj:3}
  \end{subfigure}
  \caption{Examples of projections}
  \label{fig:choice:bho}
\end{figure}
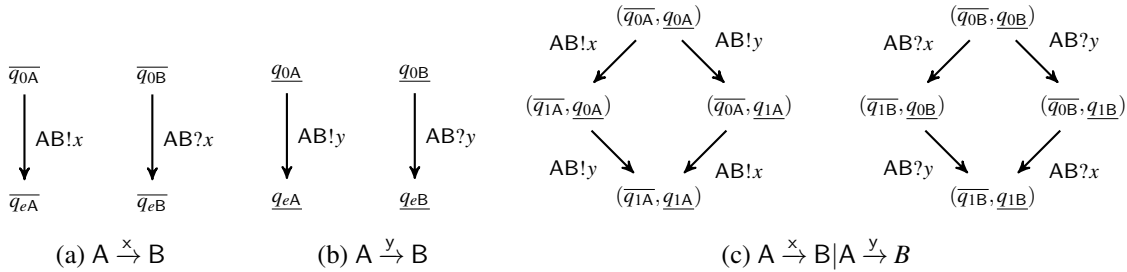

The following theorem shows that the system made of the projections of a
g-choreography $\aG$ is deadlock free if $\gsem$ is defined.
\begin{theorem}\label{thm:deadlock}
For a $\aG \in \gset$ let $\aConf_0$ be the initial state of the
communicating system 
$(min(\geproj[\aG][\ptp][q_{0\ptp}][q_{e\ptp}]))_{\ptp \in \ptpset}$. 
If $\gsem \neq \perp$ and 
$\aConf_0 \TRANSS{\ae_0} \dots \TRANSS{\ae_{n-1}}\aConf_n$ then $\aConf_n$ is not a deadlock.
\end{theorem}
\begin{proof}[Proof sketch]
The proof of the theorem is done
by structural induction over the syntax of g-choreography.
The base cases are straightforward, since the projection of a empty choreography
or of a single interaction can not lead to a deadlock.
For the inductive steps, we rely on the fact that minimisation of CFSM preserves
the language of the communicating system and does not introduce deadlocks.  
For sequential and parallel composition, the proof is done by showing that if
there is a deadlock in the composed communicating system, then there must be a
deadlock in at least one of the constituent systems. This holds
straightforwardly for the sequential composition. 
For the parallel composition, we note that
\begin{itemize}
\item in each thread, every output of a message, say $\msg$, has a corresponding input action in a receiving machine, say $\p$;
\item the machine $M_{\p}$ of the receiver $\p$ is the product of the threads on $\p$.
\end{itemize}
Therefore, the configurations where the message $\msg$ is sent have to reach a configuration where $\p$ has the reception of $\msg$ enabled (otherwise in one of the threads there would be a deadlock). Hence, eventually $\msg$ will be consumed.

For the non-deterministic composition, we show that if there is a trace in system $\aCS$ made of machines ${(\gcho[][\aG_1][\aG_2])}$ with $\p \in \ptpset$, then there must be the same trace in one of the systems made of machines $\geproj[{\aG_1}]$ or $\geproj[{\aG_2}]$. This is due to the well-branched condition.
If participant $\ptp[B]$ selects $\geproj[{\aG_i}]$ in the communicating system $\aCS$ then all other participants are forced to follow the same choice.
This allows us to build a simulation relation between the communicating system  of the non-deterministic choice and the one consisting of the CFSM
$(\geproj[{\aG_i}])_{\p \in  \ptpset}$.
\end{proof}
%For the product construction
%of CFSMs this property does not hold in general. However, it holds for the 

The following theorem shows that the traces of the system made of the
projections of a 
g-choreography $\aG$ are included in the
language of the g-choreography if $\gsem$ is defined.
\begin{theorem}\label{thm:lang}
For a $\aG \in \gset$ let $\aCS = (min(\geproj[\aG][\ptp][q_{0\ptp}][q_{e\ptp}]))_{\ptp \in \ptpset}$. 
If $\gsem \neq \perp$ then
$\rlang[\aCS] \subseteq \rlang[\aG]$.
\end{theorem}
\begin{proof}[Proof sketch]
  The proof of the theorem is done by structural induction over the
  syntax of the \linebreak g-choreographies.  The two main tasks are to show that
  ($i$) the dependencies are preserved in the case of sequential
  composition and ($ii$) no additional communication occurs in the
  case of parallel composition.  For the sequential composition we
  proceed as follows.  By definition, every word $\aW_0$ in
  $\rlang[\gseq[]]$ is the shuffling of two words, $\aW \in
  \rlang[\aG]$ and $\aW' \in \rlang[\aG']$.  Additionally, the side
  condition of the semantics of sequential composition ensures that
  all the events of $\aW$ having subject $\ptp$ precede in $\aW_0$
  every event of $\aW'$ with subject $\ptp$.  For the second task we
  rely on the fact $\gsem$ is defined and we follow the same reasoning
  done for \cref{thm:deadlock}.
\end{proof}  
In general, the converse of the inclusion in \cref{thm:lang}, that is $\rlang[\aG] \subseteq \rlang[\aCS]$, does not hold.
The reason is due to the fact that the semantics of parallel
composition of g-choreographies does not assume a FIFO policy on
channels.
In fact, the communicating system can have less behaviours than the
interleaving of the two constituent threads because of the additional
dependencies imposed by FIFO channels.
For instance, take the g-choreography
$\aG = \gpar[][{\gint[][@][x][@]}][{\gint[][@][y][\q]}]$; the word
$\aout[@][@][][{\msg[x]}]\aout[@][@][][{\msg[y]}]\ain[@][@][][{\msg[y]}]\ain[@][@][][{\msg[x]}]$ is in
$\rlang[\aG]$ but it is not in
$\rlang[{(min(\geproj[\aG][\ptp][q_{0\ptp}][q_{e\ptp}]))_{\ptp \in
    \ptpset}}]$.

\section{Conclusions}\label{conc:sec}
We introduced an abstract semantics framework of choreographies
expressed as global graphs.
Our approach is oblivious of the underlying communication semantics
and, as discussed below, can be easily adapted to alternative
semantics.
We showed that our framework is adequate by demonstrating how it can
suitably be casted in the context of communicating machines.
Our framework seems to be more expressive than existing ones; it
allows the same participant to operate in both threads of the parallel
composition and it does not force passive participants to receive a
message signalling the selected choice as first operation in a
non-deterministic composition.
This is possible due to the well-branched condition.
Interestingly, this condition is parametric and depends on the
strategy used to find the bijection required by reflection.
This can range from using always the empty bijection (thus enforcing
the same syntactical constrains of the existing proposals) to finding
a graph isomorphism.
A projection algorithm, different from the one proposed here, can
reuse the mechanism used to check the well-branched condition to
identify the common behavior of participants and avoid using
minimization.

The independence of the global semantics from the local one is evident
from \cref{thm:lang}.
We regard as a good property of our semantics the fact that global
artefacts have \quo{more executions} than the local ones obtained from
their projections.
Intuitively, this amounts to say that projections are refinements of
the (more abstract) global view.
Another advantage is that changing local artifacts does not
necessarily require to modify the semantics of the global view.
For example, if we consider CFMSs where buffers are used as
multisets (instead of as FIFO queues), then all our constructions
apply and the relation in \cref{thm:lang} is language equality
rather than just inclusion.

Our semantic framework is amenable of variations to consider
different semantics at the global level.
For instance, an alternative semantics of global views could consider
asynchronous outputs; this can be easily formalised by removing the
causal dependency between the outputs of two sequential interactions
(i.e. the topmost dotted arrows of both \cref{fig:seq:1} and
\cref{fig:seq:4} are removed).
However, this change is sound depending on the semantics of the local
artifacts.
In fact, the projections of \cref{fig:seq:4} can led to a deadlock if
the outputs are interleaved and FIFO CFMS are used as local artifacts,
while the interleaved outputs do not cause deadlocks if multiset
buffers are used by local artifacts.
As another variant one could consider a semantics where a sender
has to wait for the receiver to consume the sent message before
proceeding; this is simply attained by adding a causal dependency
from the input of $\ptp[B]$ in \cref{fig:seq:1} to the output
from $\ptp$ to $\ptp[C]$ (while removing the dotted relation).
We conjecture that this semantics would correspond to the half-duplex
semantics of CFSMs.

A distinguishing feature of this proposal is that it fixes a
specification language of global artefacts that is not a dependent
variable of the semantics of the local views.
An interesting future direction is to explore alternative projection
algorithms.
We plan to define projections that exploit reflections.
This could be better explained by observing what happens when
projecting the simple choregraphy
$\gcho[][{\gseq[][{\gint[]}][{\gint[][@][x]}]}][{\gseq[][{\gint[]}][{\gint[][@][y]}]}]$,
say on participant $\q$ (we ignore control points because immaterial).
Our algorithm yields the following machine:
\[\begin{tikzpicture}[->,>=stealth',shorten >=1pt,auto,node distance=1.7cm, thick,main node/.style={font=\scriptsize\sffamily\bfseries}]
  \node[main node] (BPS) [right =3cm of  PS]  {$q_0$};
  \node[main node] (B1)  [below left of = BPS] {$q$};
  \node[main node] (B2)  [below right of = BPS] {$q'$};
  \node[main node] (BE)  [below right of = B1] {$q_e$};

  \path[every node/.style={font=\scriptsize\sffamily}]
  (BPS) edge node {$\ain[@][@][]$} (B2)
  (B2) edge node {$\ain[@][@][][y]$} (BE)
  ;
  \path[every node/.style={font=\scriptsize\sffamily}]
  (BPS) edge node [xshift=-1cm,yshift=.5cm]  {$\ain[@][@][]$} (B1)
  (B1) edge node [xshift=-1cm,yshift=-.5cm] {$\ain[@][@][][x]$} (BE)
  ;
\end{tikzpicture}
\qquad \text{which after minimisation becomes} \qquad
\begin{tikzpicture}[->,>=stealth',shorten >=1pt,auto,node distance=1.3cm , thick,main node/.style={font=\scriptsize\sffamily\bfseries}]
  \node[main node] (B1) [right =3cm of  PS]  {$q_0$};
  \node[main node] (B2)  [below of = B1 ] {$q'$};
  \node[main node] (BE)  [below of = B2 ] {$q_e$};

  \path[every node/.style={font=\scriptsize\sffamily}]

  (B1) edge node {$\ain[@][@][]$} (B2)
  (B2) edge[bend left] node {$\ain[@][@][][x]$} (BE)
  (B2) edge[bend right] node[anchor=center, left] {$\ain[@][@][][y]$} (BE)
  ;
\end{tikzpicture}
\]
However, exploiting the bijection of the reflection, one could
directly obtain the machine on the right (avoiding the cost of
minimising machines).
Note that other projection algorithms capable of handling
the example above (as e.g., the one in~\cite{lty15}) also require
minimisation, while projections based on types (as e.g., the ones
in~\cite{honda16jacm}) are undefined on the previous example because
they require prefixes of branches to be pairwise different.

Finally, to simplify the presentation we used loop-free global
graphs. However, all results presented here can be easily extended to
graphs with structured loops that are represented as repetitions of
g-choreography. This is possible since the semantics side-conditions
do not depend on the (possibly infinite) language of the choreography,
but rather on the hypergraps, which are finite.

%%% Local Variables:
%%% mode: latex
%%% TeX-master: "main"
%%% End:

%  LocalWords:  choreographies bijection isomorphism

\nocite{*}
\bibliographystyle{eptcs}
\bibliography{bib}

\begin{thebibliography}{10}
\providecommand{\bibitemdeclare}[2]{}
\providecommand{\surnamestart}{}
\providecommand{\surnameend}{}
\providecommand{\urlprefix}{Available at }
\providecommand{\url}[1]{\texttt{#1}}
\providecommand{\href}[2]{\texttt{#2}}
\providecommand{\urlalt}[2]{\href{#1}{#2}}
\providecommand{\doi}[1]{doi:\urlalt{http://dx.doi.org/#1}{#1}}
\providecommand{\bibinfo}[2]{#2}

\bibitemdeclare{misc}{bpel}
\bibitem{bpel}
\bibinfo{author}{Charlton \surnamestart Barreto\surnameend} \&
  \bibinfo{author}{\surnamestart et~al.\surnameend} (\bibinfo{year}{2007}):
  \emph{\bibinfo{title}{Web Services Business Process Execution Language
  Version 2.0}}.
\newblock
  \bibinfo{howpublished}{\url{https://www.oasis-open.org/committees/download.php/23964/wsbpel-v2.0-primer.htm}}.

\bibitemdeclare{article}{bdft16}
\bibitem{bdft16}
\bibinfo{author}{Davide \surnamestart Basile\surnameend},
  \bibinfo{author}{Pierpaolo \surnamestart Degano\surnameend},
  \bibinfo{author}{Gian-Luigi \surnamestart Ferrari\surnameend} \&
  \bibinfo{author}{Emilio \surnamestart Tuosto\surnameend}
  (\bibinfo{year}{2016}): \emph{\bibinfo{title}{Relating two automata-based
  models of orchestration and choreography}}.
\newblock {\sl \bibinfo{journal}{JLAMP}}
  \bibinfo{volume}{85}(\bibinfo{number}{3}), pp. \bibinfo{pages}{425 -- 446},
  \doi{10.1016/j.jlamp.2015.09.011}.

\bibitemdeclare{inproceedings}{bmt14}
\bibitem{bmt14}
\bibinfo{author}{Laura \surnamestart Bocchi\surnameend},
  \bibinfo{author}{Hern{\'{a}}n~C. \surnamestart Melgratti\surnameend} \&
  \bibinfo{author}{Emilio \surnamestart Tuosto\surnameend}
  (\bibinfo{year}{2014}): \emph{\bibinfo{title}{Resolving Non-determinism in
  Choreographies}}.
\newblock In: {\sl \bibinfo{booktitle}{ESOP}}, pp. \bibinfo{pages}{493--512},
  \doi{10.1007/978-3-642-54833-8\_26}.

\bibitemdeclare{article}{bz83}
\bibitem{bz83}
\bibinfo{author}{Daniel \surnamestart Brand\surnameend} \&
  \bibinfo{author}{Pitro \surnamestart Zafiropulo\surnameend}
  (\bibinfo{year}{1983}): \emph{\bibinfo{title}{{On Communicating Finite-State
  Machines}}}.
\newblock {\sl \bibinfo{journal}{Journal of the ACM}}
  \bibinfo{volume}{30}(\bibinfo{number}{2}), pp. \bibinfo{pages}{323--342},
  \doi{10.1145/322374.322380}.

\bibitemdeclare{article}{chy07}
\bibitem{chy07}
\bibinfo{author}{Marco \surnamestart Carbone\surnameend},
  \bibinfo{author}{Kohei \surnamestart Honda\surnameend} \&
  \bibinfo{author}{Nobuko \surnamestart Yoshida\surnameend}
  (\bibinfo{year}{2007}): \emph{\bibinfo{title}{A Calculus of Global
  Interaction based on Session Types}}.
\newblock {\sl \bibinfo{journal}{Electronic Notes in Theoretical Computer
  Science}} \bibinfo{volume}{171}(\bibinfo{number}{3}), pp. \bibinfo{pages}{127
  -- 151}, \doi{10.1016/j.entcs.2006.12.041}.

\bibitemdeclare{article}{DBLP:journals/corr/abs-1203-0780}
\bibitem{DBLP:journals/corr/abs-1203-0780}
\bibinfo{author}{Giuseppe \surnamestart Castagna\surnameend},
  \bibinfo{author}{Mariangiola \surnamestart Dezani-Ciancaglini\surnameend} \&
  \bibinfo{author}{Luca \surnamestart Padovani\surnameend}
  (\bibinfo{year}{2012}): \emph{\bibinfo{title}{On Global Types and Multi-Party
  Session}}.
\newblock {\sl \bibinfo{journal}{LMCS}}
  \bibinfo{volume}{8}(\bibinfo{number}{1}), \doi{10.2168/LMCS-8(1:24)2012}.

\bibitemdeclare{article}{cdyp16}
\bibitem{cdyp16}
\bibinfo{author}{Mario \surnamestart Coppo\surnameend},
  \bibinfo{author}{Mariangiola \surnamestart
  {Dezani{-}Ciancaglini}\surnameend}, \bibinfo{author}{Nobuko \surnamestart
  Yoshida\surnameend} \& \bibinfo{author}{Luca \surnamestart
  Padovani\surnameend} (\bibinfo{year}{2016}): \emph{\bibinfo{title}{Global
  progress for dynamically interleaved multiparty sessions}}.
\newblock {\sl \bibinfo{journal}{Mathematical Structures in Computer Science}}
  \bibinfo{volume}{26}(\bibinfo{number}{2}), pp. \bibinfo{pages}{238--302},
  \doi{10.1017/S0960129514000188}.

\bibitemdeclare{inproceedings}{ivan}
\bibitem{ivan}
\bibinfo{author}{Mila \surnamestart {Dalla Preda}\surnameend},
  \bibinfo{author}{Maurizio \surnamestart Gabbrielli\surnameend},
  \bibinfo{author}{Saverio \surnamestart Giallorenzo\surnameend},
  \bibinfo{author}{Ivan \surnamestart Lanese\surnameend} \&
  \bibinfo{author}{Mauro \surnamestart Jacopo\surnameend}
  (\bibinfo{year}{2015}): \emph{\bibinfo{title}{Dynamic Choreographies - Safe
  Runtime Updates of Distributed Applications}}.
\newblock In: {\sl \bibinfo{booktitle}{{COORDINATION} 2015}}, pp.
  \bibinfo{pages}{67--82}, \doi{10.1007/978-3-319-19282-6\_5}.

\bibitemdeclare{inproceedings}{dy12}
\bibitem{dy12}
\bibinfo{author}{{Pierre-Malo} \surnamestart Deni{\'e}lou\surnameend} \&
  \bibinfo{author}{Nobuko \surnamestart Yoshida\surnameend}
  (\bibinfo{year}{2012}): \emph{\bibinfo{title}{{Multiparty Session Types Meet
  Communicating Automata}}}.
\newblock In: {\sl \bibinfo{booktitle}{ESOP}}, pp. \bibinfo{pages}{194--213},
  \doi{10.1007/978-3-642-28869-2\_10}.

\bibitemdeclare{article}{honda16jacm}
\bibitem{honda16jacm}
\bibinfo{author}{Kohei \surnamestart Honda\surnameend}, \bibinfo{author}{Nobuko
  \surnamestart Yoshida\surnameend} \& \bibinfo{author}{Marco \surnamestart
  Carbone\surnameend} (\bibinfo{year}{2016}): \emph{\bibinfo{title}{Multiparty
  Asynchronous Session Types}}.
\newblock {\sl \bibinfo{journal}{Journal of the ACM}}
  \bibinfo{volume}{63}(\bibinfo{number}{1}), pp. \bibinfo{pages}{9:1--9:67},
  \doi{10.1145/2827695}.
\newblock \bibinfo{note}{Extended version of a paper presented at {POPL08}}.

\bibitemdeclare{misc}{w3c:cho}
\bibitem{w3c:cho}
\bibinfo{author}{Nickolas \surnamestart Kavantzas\surnameend},
  \bibinfo{author}{Davide \surnamestart Burdett\surnameend},
  \bibinfo{author}{Gregory \surnamestart Ritzinger\surnameend},
  \bibinfo{author}{Tony \surnamestart Fletcher\surnameend} \&
  \bibinfo{author}{Yves \surnamestart Lafon\surnameend} (\bibinfo{year}{2004}):
  \emph{\bibinfo{title}{Web Services Choreography Description Language Version
  1.0}}.
\newblock
  \bibinfo{howpublished}{\url{http://www.w3.org/TR/2004/WD-ws-cdl-10-20041217}}.

\bibitemdeclare{inproceedings}{lty15}
\bibitem{lty15}
\bibinfo{author}{Julien \surnamestart Lange\surnameend},
  \bibinfo{author}{Emilio \surnamestart Tuosto\surnameend} \&
  \bibinfo{author}{Nobuko \surnamestart Yoshida\surnameend}
  (\bibinfo{year}{2015}): \emph{\bibinfo{title}{{From Communicating Machines to
  Graphical Choreographies}}}.
\newblock In: {\sl \bibinfo{booktitle}{POPL15}}, pp. \bibinfo{pages}{221--232},
  \doi{10.1145/2676726.2676964}.

\end{thebibliography}

\end{document}

\end{document}

%%% Local Variables: 
%%% mode: latex
%%% TeX-master: t
%%% End: 

%  LocalWords:  choreographies Hypergraphs